\documentclass[1p]{elsarticle}

\usepackage{amsmath}
\usepackage{amssymb}
\usepackage{amsfonts}
\usepackage{amstext}
\usepackage{graphicx}
\usepackage{amscd}
\usepackage{mathrsfs}
\usepackage{stmaryrd}
\usepackage{slashed}
\usepackage{color}
\usepackage{upgreek}
\usepackage{tikz-cd}
\usepackage{hhline}
\usepackage{amsbsy}

\newcommand{\RE}{\Re\mathrm{e}}
\newcommand{\IM}{\Im\mathrm{m}}

\newcommand{\Peg}{\underset{\leftarrow}{\mathbb{P}e}}
\newcommand{\llangle}{\langle \hspace{-0.2em} \langle}
\newcommand{\rrangle}{\rangle \hspace{-0.2em} \rangle}

\newcommand{\rrrangle}{\rangle \hspace{-0.2em} \rangle \hspace{-0.2em} \rangle}

\newcommand{\id}{\mathrm{id}}
\newcommand{\Env}{\mathrm{Env}}

\newcommand{\Lin}{\mathrm{Lin}}

\newcommand{\tr}{\mathrm{tr}}

\newcommand{\Ran}{\mathrm{Ran}}
\newcommand{\odd}{\mathrm{odd}_\beta}
\newcommand{\dom}{\mathrm{dom}}

\newtheorem{defi}{Definition}
\newtheorem{theo}{Theorem}
\newtheorem{prop}{Property}

\newenvironment{proof}{\noindent \textit{Proof:}\small }{\normalsize \hfill $\Box$ \\}

\begin{document}
\begin{frontmatter}

\title{Quasicoherent states of noncommutative D2-branes, Aharonov-Bohm effect and quantum M\"obius strip}

\author[uti]{David Viennot}
\address[uti]{Universit\'e Marie \& Louis Pasteur, Institut UTINAM (CNRS UMR 6213), Observatoire de Besan\c con), 41bis Avenue de l'Observatoire, BP 1615, 25010 Besan\c con cedex, France.}

\begin{abstract}
We find an analytical formula for the quasicoherent states of 3D fuzzy spaces defined by algebras generated by bosonic creation and annihilation operators. This one is expressed in a representation onto the coherent states of the CCR algebra. Such a fuzzy space can be assimilated to a noncommutative D2-brane of the M-theory (but also as a model of a qubit in contact with a bosonic environment). We apply this formula onto a D2-brane wrapped around an axis to obtain the geometry of a noncommutative cylinder. We show that the adiabatic transport of its quasicoherent states exhibits a topological effect similar to the Aharonov-Bohm effect. We study also a D2-brane wrapped and twisted to have the geometry of a noncommutative M\"obius strip. Finally we briefly present the other two examples of a noncommutative torus and of a noncommutative Klein bottle.
\end{abstract}

\end{frontmatter}

\tableofcontents

\section{Introduction}
Fuzzy spaces \cite{Barrett} are special cases of Connes' noncommutative geometry \cite{Connes}. 3D fuzzy spaces are models of noncommutative D2-branes in M-theory in the context of the Banks-Fischler-Shenker-Susskind (BFSS) matrix theory \cite{Banks,Klammer,Steinacker,Kunter,Sahakian} (the spacetime dimension being reduced from 9+1 dimensions to 3+1 by orbifoldisation \cite{Berenstein}). We can also consider 3D fuzzy spaces to model the control of a qubit entangled with a large environment \cite{Viennot1}. Quasicoherent states are an important tool to study the geometry of a fuzzy space \cite{Schneiderbauer,Steinacker2}. They generalise the Perelomov coherent states of Lie algebras \cite{Perelomov} to fuzzy spaces, as being states minimising the Heisenberg uncertainties concerning the geometric observables. Moreover, they play an important role in the interpretation of brane quantum gravity at the adiabatic limit \cite{Viennot1,Viennot2}. Unfortunately, the explicit computation of the quasicoherent states is not obvious. In this paper we focus onto 3D fuzzy spaces described by an operator algebra defined with bosonic creation and annihilation operators $a$ and $a^+$. We refer to such a geometric structure as a CCR (canonical commutation relation) D2-brane. We show that in this case, we can obtain an analytical formula for the quasicoherent states based on a decomposition onto the Perelomov coherent states of the CCR algebra generated by $\{a,a^+,\id\}$. The set of the CCR coherent states being uncountable overcomplete, this decomposition appears as an integration in the complex plane, but in return any operator of the Fock space has a diagonal representation onto the coherent state set \cite{Perelomov}.\\

We are in particular interested by the quasicoherent states of the 2D-branes which are noncommutative counterparts of the classical topologically flat surfaces (plane, cylinder, torus, M\"obius strip, Klein bottle, real projective plane). These classical surfaces can be built from the plane by an operation of quotient by a lattice with eventually twists, following by an embedding (or an immersion) in $\mathbb R^3$. We want to consider D2-branes built with a similar way. We consider then first the topological noncommutative plane described by the $*$-algebra $\mathfrak X = \Env(a,a^+)$ of all polynomials of $a$ and $a^+$. The equivalent operation of the quotient/twist consists to a restriction to a $*$-subalgebra $\mathfrak X_0$ modelled onto the example of the algebra of analytical functions of the classical topological surface. The equivalent operation of the embedding/immersion consists to endow the D2-brane by a Dirac operator of the BFSS matrix theory. Such an operator defines the metric properties of the D2-brane consistent with the embedding/immersion of its average surface in the quasicoherent states \cite{Viennot2}. An interesting question is to find the purely topological manifestations of such a kind of D2-brane. Since the geometric properties of a D2-brane is intimately linked to the adiabatic transport of their quasicoherent states \cite{Viennot1,Viennot2}, we can endow a D2-brane with a Berry potential (generated a Berry phase) and the associated Berry curvature \cite{Berry, Simon, Shapere, Bohm}. It is well known that these Berry potential and curvature are similar to magnetic potential and field. We can then think these quantities as intrinsic magnetic potential and magnetic field of the D2-brane. For the case of the noncommutative cylinder, the operations of quotient and embedding can be assimilated to wrap the noncommutative plane around an axis in the space. A natural question is then the following: its intrinsic magnetic potential does exhibit an Aharonov-Bohm effect \cite{Aharonov1, Aharonov2} associated with this topology? Some works consider Aharonov-Bohm effect in noncommutative spaces \cite{Chaichian, Chaichian2, Li, Anacleto, Jing}, essentially by considering gauge theory in noncommutative phase space endowed with a Moyal product. Our approach is different. We do not consider an ``external'' magnetic field onto a topologically trivial noncommutative space, but the ``inner'' magnetic field (i.e. the Berry connection) onto a non topologically trivial noncommutative D2-brane. We will see that such a situation exhibits a kind of Aharonov-Bohm effect.\\

We can briefly recall here the principle of the Aharonov-Bohm effect to fix the notations and introduce some reference equations: Let a charged particle be transported along a closed path $\mathscr C$ in a space where a magnetic field lives. Let $\phi$ be the initial wave function of the particle supposed to be a wave packet localised around a point $\vec x_0 \in \mathscr C$. The transported wave function is:
\begin{equation}
  \psi(\vec x) = e^{\imath q \oint_{\mathscr C} \mathrm A} \phi(\vec x)
\end{equation}
where $q$ is the particle electric charge and $\mathrm A \in \Omega^1 \mathbb R^3$ is the magnetic potential (represented by a differential 1-form: $\mathrm A = \vec {\mathrm A} \cdot d\vec \ell$ where $d\vec \ell$ is the vector of the infinitesimal variation of the position). If the magnetic field is induced by an infinite solenoid of axis $\vec e_z$ and radius $R$, then ($r$ is the distance to $z$-axis):
\begin{equation}
  \mathrm A = \begin{cases} \kappa \frac{r^2}{R^2} d\theta & \text{if } r \leq R \\ \kappa d\theta & \text{if } r \geq R \end{cases}
\end{equation}
where $\kappa = \frac{\nu I R^2}{2}$ ($\nu$ being the number of coils by length unit and $I$ is the current intensity of the solenoid). The magnetic field $\mathrm B = d\mathrm A$ is then $\mathrm B = 2 \frac{\kappa}{R^2} rdr \wedge d\theta$ ($\vec {\mathrm B} = 2 \frac{\kappa}{R^2} \vec e_z$) for $r < R$ and $\mathrm B = 0$ for $r>R$. If the path $\mathscr C$ is inner the solenoid, the quantum phase of the transport is $\oint_{\mathscr C} \mathrm A = \iint_{\mathscr S} \mathrm B$ where $\mathscr S$ is the minimal surface delimited by $\mathscr C$ (if $\mathscr C$ is in a plane normal to $\vec e_z$, $\oint_{\mathscr C} \mathrm A = \frac{2\kappa}{R^2} \mathcal A_{\mathscr S}$ where $\mathcal A_{\mathscr S}$ is the area of $\mathscr S$). The quantum phase is geometric in the sense of that it depends on the shape of $\mathscr C$. In contrast, if $\mathscr C$ is outer the solenoid, the quantum phase is $\oint_{\mathscr C} \mathrm A = n \kappa$ where $n$ the number of turns of $\mathscr C$ around the solenoid. The quantum phase is now topological, it does not depend on the shape of $\mathscr C$ (it is conserved by continuous deformations of $\mathscr C$) but only on the number of turns around $\vec e_z$. If $\kappa \not\in \frac{2\pi}{q} \mathbb Z$, this topological phase is not trivial and can be revealed by interferometry whereas the particle see a zero magnetic field throughout its transport. This is the Aharonov-Bohm effect. The effect can be extended in the case where $R \to 0$ with $\kappa$ kept constant (with $\nu$ or $I$ $\to \infty$).\\

The Berry phase is a related phenomenon associated with the adiabatic transport of any quantum system obeying to a Schr\"odinger equation. Let $H(x)$ be a self-adjoint  $x$-parameter dependent Hamiltonian of the quantum system, $\lambda(x)$ be an eigenvalue (supposed non-degenerate almost everywhere) and $|\lambda(x)\rangle$ be the associated normalised eigenvector. If the parameters are slowly varied with the time $[0,T] \ni t \mapsto x(t) \in M$, drawing a closed path $\mathscr C$ in the parameter space $M$, then the solution of the Schr\"odinger equation can be approximated by \cite{Teufel}
\begin{equation}
  |\psi(T)\rangle \simeq e^{-\imath \int_0^T \lambda(x(t))dt} e^{-\imath \oint_{\mathscr C} \mathrm A} |\lambda(x(0))\rangle
\end{equation}
if $|\psi(0)\rangle = |\lambda(x(0))\rangle$; where $\mathrm A = -\imath \langle \lambda(x)|d|\lambda(x) \rangle \in \Omega^1M$ generates a geometric phase called the Berry phase ($d$ being the exterior differential of $M$). $\mathrm A$ is similar to a magnetic potential living in the parameter space $M$, with the associated magnetic field $\mathrm B=d\mathrm A = -\imath d\langle \lambda(x)| \wedge d|\lambda(x)\rangle \in \Omega^2M$. The sources of this magnetic field are identified as magnetic monopoles living on $M$ (or in an extension of $M$) at points $x_*$ where $\lambda(x_*)$ crosses another eigenvalue of $H(x_*)$. It is with this Berry potential associated with the quasicoherent states that we will study a possible Aharonov-Bohm effect onto a wrapped noncommutative D2-brane.\\

This paper is organised as follows. Section 2 introduces the precise definition of a D2-brane, of its quasicoherent states and some related important properties. Section 3 presents the main theorem of this paper describing the quasicoherent states of a CCR D2-brane as a superposition of Perelomov coherent states of the CCR algebra. We use this one to compute the general analytical formula of the magnetic potential of a CCR D2-brane. Section 4 presents the case of the noncommutative cylinder and shows that its magnetic potential can be decomposed into two parts, a first one generating a geometric phase and a second generating a topological phase similar to the one of the Aharonov-Bohm effect. Section 5 presents the other example of the noncommutative M\"obius strip for which the topological effects due to the twist (and the resulting non-orientability) are studied. Finally section 6 briefly presents other two examples: the noncommutative torus and the noncommutative Klein bottle.\\

\textit{Throughout this paper, we use the Planck units ($\ell_P=m_P=t_P=1$, $\hbar=c=G=1$) for the applications in M-theory, and the atomic units ($\hbar = \frac{e^2}{4\pi \epsilon_0}=1$) for the applications in quantum information theory.}

\section{Noncommutative D2-branes and their quasicoherent states}
\subsection{Definition of a noncommutative D2-brane}
\begin{defi}[Noncommutative D2-brane]
  A noncommutative D2-brane (so-called 3D fuzzy space) is a noncommutative manifold defined by a spectral triple $\mathfrak M = (\mathfrak X,\mathbb C^2\otimes \mathscr F,\slashed D_x)$ where
  \begin{itemize}
  \item $\mathscr F$ is a separable Hilbert space.
  \item $\mathfrak X$ is an $*$-algebra of operators acting on $\mathscr F$.
  \item $\slashed D_x = \sigma_i \otimes (X^i-x^i)$ is the Dirac operator of the noncommutative manifold where $(\sigma_i)$ are the Pauli matrices, $(X^i)$ are self-adjoint operators of $\mathfrak X$ and $x\in \mathbb R^3$ is a classical parameter.
  \end{itemize}
\end{defi}
A D2-brane can be viewed as a quantum extended object. $(X^i)$ are the quantum observables of the coordinates onto $\mathfrak M$. In general these ones do not commute $[X^i,X^j] = \imath \Theta^{ij} \not=0$. $\mathscr F$ is the Hilbert space of ``states of location'' onto $\mathfrak M$. $(\sigma_i)$ are the observable of local orientation onto $\mathfrak M$. If we see $\mathfrak M$ as a noncommutative surface, $\vec \sigma$ is the quantum normal vector to $\mathfrak M$ (see \cite{Viennot1}). For a state $\Psi \in \mathbb C^2 \otimes \mathcal H$, $\llangle \Psi|\vec X|\Psi \rrangle$ is the mean location in the state $\Psi$, and $\llangle \Psi|\vec \sigma|\Psi \rrangle$ is the mean normal vector in the state $\Psi$. $\slashed D_x$ is the fundamental observable of $\mathfrak M$ (the observable defining its geometry), representing the minimal coupling between the quantum degrees of freedom of location and of orientation. $(x^i)$ are the probe classical coordinates. $x$ is the location of the probe of the observer in the classical space $\mathbb R^3$ of this one.\\

This structure can model two interesting physical situations:
\begin{itemize}
\item In M-theory, $\slashed D_x$ is the Dirac operator in the BFSS matrix model \cite{Banks,Klammer,Steinacker,Kunter,Sahakian}, for a massless fermionic string (of spin $\vec \sigma$) linking a noncommutative D2-brane ($X^i$) and a probe D0-brane ($x^i$); the space dimension being reduced from 9 to 3 by orbifoldisation \cite{Berenstein}. $\slashed D_x$ is the displacement energy of the system (the ``tension'' energy of the fermionic string).
\item In quantum information theory, $\slashed D_x$ is the interaction Hamiltonian of a qubit $(\sigma_i)$ in contact with a large environment $(X^i)$ responsible of decoherence phenomena induced by entanglement. $-x^i \sigma_i$ is the part corresponding to the control of the qubit (the observer by making variations $t \mapsto x^i(t)$ try to realise quantum logical operations onto the qubit).
\end{itemize}

The Dirac operator can be rewritten in the canonical basis of $\mathbb C^2$ as
\begin{equation}
  \slashed D_x = \left(\begin{array}{cc} X^3-x^3 & A^\dagger - \bar \alpha \\ A - \alpha & -(X^3-x^3) \end{array} \right)
\end{equation}
where $A = X^1+\imath X^2$ and $\alpha = x^1+\imath x^2 \in \mathbb C$.\\

Since the coordinates observables do not commute, they are subject to a Heisenberg uncertainty relation:
\begin{prop}[Heisenberg uncertainty relation]
  Let $\mathfrak M$ be a D2-brane with $[X^i,X^j] = \imath \Theta^{ij}$, and $\Delta \vec X^2 \equiv \llangle \Psi |\vec X^2 |\Psi \rrangle - \llangle \Psi|\vec X|\Psi \rrangle^2$ be the square uncertainty onto the location in the state $\Psi \in \mathbb C^2 \otimes \mathscr F$ ($\vec X^2 \equiv \delta_{ij} X^i X^j$). We have
  \begin{equation}
    \Delta \vec X^2 \geq \frac{1}{2} {\varepsilon_{ij}}^k \llangle \sigma_k \otimes \Theta^{ij} \rrangle
  \end{equation}
\end{prop}
We can note that by construction, the Heisenberg uncertainty relation for a D2-brane depends on the entanglement of the considered state $\Psi \in \mathbb C^2 \otimes \mathcal H$.\\

\begin{proof}
  A direct calculation shows that
  \begin{equation}
    \slashed D_x^2 = (\vec X-\vec x)^2 + \frac{\imath}{2} {\varepsilon_{ij}}^k \sigma_k \otimes [X^i,X^j]
  \end{equation}
  We have then for the particular case $\vec x = \llangle \vec X \rrangle$
  \begin{eqnarray}
    & & \slashed D^2_{\llangle X \rrangle} = (\vec X - \llangle \vec X \rrangle)^2 + \frac{\imath}{2} {\varepsilon_{ij}}^k \sigma_k \otimes [X^i,X^j] \\
    & \Rightarrow & \vec X^2 -2 \llangle \vec X \rrangle \cdot \vec X + \llangle \vec X \rrangle^2 = \slashed D^2_{\llangle X \rrangle} + \frac{1}{2} {\varepsilon_{ij}}^k \sigma_k \otimes \Theta^{ij}
  \end{eqnarray}
  It follows that
  \begin{equation}\label{eqUncertainty}
    \llangle \vec X^2 \rrangle - \llangle \vec X \rrangle^2 = \llangle \slashed D^2_{\llangle X \rrangle} \rrangle + \frac{1}{2} {\varepsilon_{ij}}^k \llangle \sigma_k \otimes \Theta^{ij} \rrangle
  \end{equation}
  $\slashed D_x^2$ being a positive operator, $\llangle \slashed D_x^2 \rrangle \geq 0$. 
\end{proof}

\subsection{Quasicoherent states}
In Connes' noncommutative geometry, the notion of point of the classical geometry is replaced by the notion of state. But a state cannot be assimilated to a point since it is non-local: in general $\langle \psi|\phi\rangle \not=0$ even if $\psi,\phi \in \mathscr F$ are linearly independent; i.e. two independent pure states are not geometrically separable, in contrast with the classical distributions $\delta(x-x_0)$ and $\delta(x-x_1)$ in the classical case.\\
But Dirac distributions (which are assimilated to the points onto they are centred) are the distributions of the classical geometry with minimal dispersion (this one being 0). In quantum geometry, due to the Heisenberg uncertainty relation, the dispersion cannot reach 0, but we can consider states having the minimal (non-zero) dispersion as the better equivalents to the classical points. The states closest to classical points are then the quasicoherent states:
\begin{defi}[Quasicoherent state]
  Let $\mathfrak M$ be a D2-brane. Let $M_\Lambda = \{\vec x \in \mathbb R^3, \ker(\slashed D_x)\not=0\}$ be the eigenmanifold of $\mathfrak M$. $M_\Lambda \ni x \mapsto |\Lambda(x) \rrangle \in \ker \slashed D_x$ is said a quasicoherent state ($|\Lambda(x)\rrangle$ being supposed normalised).
\end{defi}

\begin{prop}
  A quasicoherent state $|\Lambda(x)\rrangle$ of $\mathfrak M$ satisfies the following properties:
  \begin{enumerate}[1.]
  \item $\llangle \Lambda(\vec x)|\vec X|\Lambda(\vec x)\rrangle = \vec x$ (the mean values of the coordinate observables span the whole of $M_\Lambda$);
  \item $\llangle \Lambda(\vec x)|\vec \sigma|\Lambda(\vec x)\rrangle \in N_xM_\Lambda$ (the mean values of the normal vector observable is a normal vector of the eigenmanifold at $x$);
  \item $\Delta_x \vec X^2 = \frac{1}{2} {\varepsilon_{ij}}^k \llangle \Lambda(x)|\sigma_k \otimes \Theta^{ij}|\Lambda(x)\rrangle$ (the Heisenberg uncertainty is minimised).
  \end{enumerate}
\end{prop}
Since the quasicoherent states minimise the Heisenberg uncertainty, they are states closest to the classical notion of point. Moreover, $M_\Lambda$ being generated by $\llangle \Lambda(\vec x)|\vec X|\Lambda(\vec x)\rrangle$, it is the classical manifold closest to $\mathfrak M$ (the classical manifold which has the geometry closest to the one of $\mathfrak M$).\\

\begin{proof}
  A direct calculation shows that $\frac{1}{2}(\sigma^i\slashed D_x+\slashed D_x\sigma^i) = X^i-x^i$, inducing that $\llangle \Lambda(x)|(X^i-x^i)|\Lambda(x)\rrangle = 0$ (since $\slashed D_x|\Lambda(x)\rrangle=0$). This proves the point 1.\\
  Let $d$ be the exterior derivative onto $M_\Lambda$. $\slashed D_x |\Lambda(x)\rrangle= 0 \Rightarrow -\sigma_idx^i|\Lambda(x)\rrangle + \slashed D_x d|\Lambda(x)\rrangle=0$. It follows that $\llangle \Lambda(x)|\sigma_i|\Lambda(x)\rrangle dx^i = 0$, and so $\llangle \Lambda(x)|\sigma_i|\Lambda(x) \rrangle \frac{\partial x^i}{\partial s^a} = 0$ ($\forall a \in \{1,2\}$) with $(s^1,s^2)$ a local curvilinear coordinate system onto $M_\Lambda$. This proves the point 2.\\
  By applying eq.(\ref{eqUncertainty}) with $|\Lambda(x)\rrangle$, we have $\Delta_x \vec X^2 = \llangle \Lambda(x)|\slashed D^2_{\llangle X\rrangle} |\Lambda(x)\rrangle + \frac{1}{2} {\varepsilon_{ij}}^k \llangle \Lambda(x)|\sigma_k \otimes \Theta^{ij}|\Lambda(x)\rrangle$. But from the point 1, we have $\llangle \Lambda(x)|\slashed D^2_{\llangle X\rrangle} |\Lambda(x)\rrangle = \llangle \Lambda(x)|\slashed D^2_{x} |\Lambda(x)\rrangle = 0$. This proves the point 3.
\end{proof}

\subsection{Adiabatic dynamics}
The time evolution of a state $|\Psi \rrangle \in \mathbb C^2 \otimes \mathscr F$ is governed by the following equation:
\begin{equation}
  \imath |\dot \Psi \rrangle = \slashed D_{x(t)} |\Psi(t) \rrangle
\end{equation}
which is the Dirac equation in the Weyl representation of a massless fermionic string linking $\mathfrak M$ and $\vec x$ or which is the Schr\"odinger equation of a qubit in contact with the environment described by $\mathfrak X$ and controlled by the classical parameters $x(t)$. We suppose that $\vec X$ is time-independent (see section \ref{conclusion} the discussion in conclusion for a generalisation to time-dependent cases). The adiabatic regime (when $x$ is slowly modified) plays an important role. In M-theory, a spacetime geometry foliated in time by space leafs diffeomorphic to $M_\Lambda$ emerges at the adiabatic limit \cite{Viennot2}. In quantum information theory, adiabatic control of qubits is a strategy of quantum computing \cite{Albash}. The adiabatic transported state is then if $|\Psi(0)\rrangle = |\Lambda(x(0))\rrangle$:
\begin{equation}
  |\Psi(T)\rrangle = e^{-\imath \oint_{\mathscr C} \mathrm A} |\Lambda(x(0))\rrangle
\end{equation}
where $\mathscr C$ is the closed path drawn on $M_\Lambda$ in $\mathbb R^3$ by the slow variation of the probe $[0,T] \ni t \mapsto x(t) \in M_\Lambda$. The geometric phase is generated by the magnetic potential $\mathrm A = -\imath \llangle \Lambda(x)|d|\Lambda(x) \rrangle \in \Omega^1M_\Lambda$ living on $M_\Lambda$. 

\section{CCR noncommutative D2-branes}
\subsection{CCR D2-branes and $|\alpha\rangle$-representation}
We said that a D2-brane $\mathfrak M$ is CCR (canonical commutation relation), if $\mathscr F$ is a bosonic Fock space (with a single mode to simplify the discussion, see section \ref{conclusion} the discussion in conclusion for multi-mode cases) and if $\mathfrak X$ is generated by two operators $U^1,U^2 \in \Env(a,a^+)$ where $a$ and $a^+$ are the bosonic annihilation and creation operators $([a,a^+]=1$). $(U^1,U^2$) play the role of local curvilinear coordinates observables of $\mathfrak M$, whereas $(X^1,X^2,X^3)$ play the role of coordinates observables of the embedding of $\mathfrak M$ in $\mathbb R^3$. The quantum equivalent of the embedding of a surface $f : M \to \mathbb R^3$, with $f(u) = f^i(u^1,u^2) \vec e_i$ ($(u^1,u^2)$ being curvilinear coordinates on $M$), is then the Dirac operator $\slashed D_x = \sigma_i \otimes (f^i(U^1,U^2)-x^i(u^1,u^2))$ with $X^i=f^i(U^1,U^2) \in \mathfrak X$ and $(u^1,u^2)$ curvilinear coordinates on $M_\Lambda$.\\
These models of D2-branes seem natural in M-theory where a D2-brane is a non-perturbative version of a quantum graviton field. In quantum information, these models correspond to a qubit in contact with a reservoir of bosons.\\

Let $|\alpha \rangle$ ($\alpha \in \mathbb C$) be a Perelomov coherent state of the boson field \cite{Perelomov}:
\begin{equation}
  |\alpha \rangle = e^{-|\alpha|^2/2} \sum_{n=0}^\infty \frac{\alpha^n}{\sqrt{n!}} |n \rangle
\end{equation}
where $(|n\rangle)_{n \in \mathbb N}$ is the canonical basis of $\mathscr F$ (boson numbers basis). $|\alpha \rangle$ is eigenvector of the annihilation operator: $a|\alpha \rangle = \alpha |\alpha \rangle$. Any state $|\psi \rangle \in \mathscr F$ can be represented onto the coherent state set \cite{Perelomov}:
\begin{equation}
  |\psi \rangle = \int_{\mathbb C} \langle \alpha | \psi \rangle |\alpha \rangle \frac{d^2\alpha}{\pi}
\end{equation}
where $d^2\alpha = d\RE(\alpha)d\IM(\alpha)$; with $\forall |\psi\rangle,|\phi\rangle \in \mathscr F$:
\begin{equation}
  \langle \phi | \psi \rangle = \int_{\mathbb C} \langle \phi|\alpha \rangle \langle \alpha | \psi \rangle \frac{d^2\alpha}{\pi}
\end{equation}
this last equation can be surprising since $\langle \beta|\alpha\rangle = e^{-\frac{1}{2}|\beta-\alpha|^2}e^{\imath \IM(\bar\beta \alpha)}$, but is due to the fact that the coherent state set is overcomplete \cite{Perelomov}. For the same reason, any operator $X \in \mathcal L(\mathscr F)$ has a diagonal representation in the $|\alpha\rangle$-basis (Sudarshan-Mehta theorem \cite{Sudarshan,Mehta}):
\begin{equation}
  X = \int_{\mathbb C} \varphi_X(\alpha) |\alpha \rangle \langle \alpha| \frac{d^2\alpha}{\pi}
\end{equation}
where $\varphi_X$ is a function or a distribution on $\mathbb C$. \ref{alphaRep} presents more details about this diagonal representation.\\

We can extend the $|\alpha\rangle$-representation.  Let $\mathscr F_0$ be the subset of $\mathscr F$ defined by
$$ \mathscr F_0 = \left\{\int_{\mathbb C} g(\alpha) |\alpha \rangle \frac{d^2\alpha}{\pi},\, \forall g \in \mathbb C^{\mathbb C} \text{ of bounded support} \right\} $$
$\mathscr F_0$ is dense in $\mathscr F$ for the strong topology \cite{Garding}. Let $\mathscr F_\infty$ be the weak topological closure of $\mathscr F_0$:
$$ \mathscr F_\infty = \left\{|\psi_\infty \rangle \text{ s.t. } \exists (\psi_n)_{n \in \mathbb N} \in (\mathscr F_0)^{\mathbb N}, \lim_{n \to +\infty} |\langle \phi|\psi_n-\psi_\infty \rangle|=0, \forall \phi \in \mathscr F \right\} $$
In general $|\psi_\infty\rangle$ is an anti-linear distribution onto $\mathbb C$. $\mathscr F_0 \subset \mathscr F \subset \mathscr F_\infty$ is a Gelfand triple ($\mathscr F_\infty$ is a rigged Hilbert space) \cite{Gadella, Madrid}. In particular, $\mathscr F_\infty$ includes states as $|\psi\rangle = \int_{\mathbb C} \psi(\alpha) |\alpha \rangle \frac{d^2\alpha}{\pi}$ but where $\psi(\alpha) \not= \langle \alpha|\psi\rangle$ (for example if $\psi(\alpha)$ is not analytical). $\psi(\alpha)$ is not the inner product of a coherent state with a state of $\mathscr F$, but by an abusive notation, we will write in the next $\psi(\alpha) \equiv \langle \alpha|\psi\rangle$ but which must be interpreted as the evaluation of the complex variable function $\psi$ at $\alpha$ and not as an inner product in $\mathscr F$.

\subsection{Quasicoherent states of a CCR D2-brane}
We can use the $|\alpha\rangle$-representation of $\slashed D_x$ to find its quasicoherent states in the rigged Fock space $\mathscr F_\infty$. We suppose that $\slashed D_x$ is self-adjoint onto $\mathscr F_0$. Since $\mathscr F_0 \subset \dom (\slashed D_x)$ \cite{Garding} and $\mathscr F_0$ is dense in $\mathscr F$, $\mathscr F_0$ is a core for $\slashed D_x$ \cite{RS}. 

\begin{theo}[Quasicoherent states of a CCR D2-brane]\label{Theo}
  Let $\slashed D_x = \left(\begin{array}{cc} X^3-x^3 & A^\dagger - \bar \alpha \\ A-\alpha & -(X^3-x^3) \end{array}\right)$ be a Dirac operator of a CCR D2-brane $\mathfrak M$ ($A,X^3 \in \Env(a,a^+)$). Let $\varphi_A,\varphi_{X^3} \in \mathbb C^{\mathbb C}$ be the diagonal $|\alpha\rangle$-representations of $A$ and $X^3$. $\ker \slashed D_x \not= \{0\}$ if and only if
  \begin{equation}
    \begin{cases} \alpha \in \Ran \varphi_A & \\ x^3 = \varphi_{X^3} \circ \varphi_A^{-1}(\alpha) & \end{cases}
  \end{equation}
  In other words, the eigensurface is $M_\Lambda = \{(\RE(\varphi_A(\beta)),\IM(\varphi_A(\beta)),\varphi_{X^3}(\beta))\}_{\beta \in \mathbb C}$.\\

  Let $\alpha_A = \varphi_A^{-1}(\alpha)$, the quasicoherent states of $\mathfrak M$ are in $\mathscr F_\infty$:
  \begin{eqnarray}
    |\Lambda(x)\rrangle & = & \int_{\mathbb C} \frac{\beta\langle \beta+\alpha_A|\alpha_A\rangle}{|\Delta \varphi_A(\beta)|^2+|\Delta \varphi_{X^3}(\beta)|^2} \left(\begin{array}{c} \overline{\Delta \varphi_A(\beta)} \\ -\Delta \varphi_{X^3}(\beta) \end{array} \right) \otimes |\alpha_A+\beta\rangle \frac{d^2\beta}{\pi N_x} \\
   |\Lambda_*(x)\rrangle & = & \int_{\mathbb C} \frac{\beta\langle \beta+\alpha_A|\alpha_A\rangle}{|\Delta \varphi_A(\beta)|^2+|\Delta \varphi_{X^3}(\beta)|^2} \left(\begin{array}{c} \Delta \varphi_{X^3}(\beta) \\ \Delta \varphi_A(\beta) \end{array} \right) \otimes |\alpha_A+\beta\rangle \frac{d^2\beta}{\pi N_x} 
  \end{eqnarray}
  where $\Delta \varphi_A(\beta) = \varphi_A(\alpha_A+\beta)-\varphi_A(\alpha_A)$ and $\Delta \varphi_{X^3}(\beta) = \varphi_{X^3}(\alpha_A+\beta)-\varphi_{X^3}(\alpha_A)$; the normalisation factor being
  \begin{equation}
    N_x^2 = \int_{\mathbb C} \frac{|\beta|^2 e^{-|\beta|^2}}{|\Delta \varphi_A(\beta)|^2+|\Delta \varphi_{X^3}(\beta)|^2} \frac{d^2\beta}{\pi}
  \end{equation}
\end{theo}

Note that $\alpha_A = \varphi^{-1}(\alpha)$ is not necessarily unique, and we have a different couple of quasicoherent states $(|\Lambda\rrangle,|\Lambda_*\rrangle)$ for each element of the preimage of $\alpha$.\\

\begin{proof}
  We write $|\Lambda \rrangle = (|\Lambda^0\rangle, |\Lambda^1\rangle)$ in the canonical basis of $\mathbb C^2$. $\slashed D_x|\Lambda \rrangle = 0 \Rightarrow$
  \begin{equation}
    \begin{cases} (X^3-x^2)|\Lambda^0\rangle + (A^\dagger - \bar \alpha)|\Lambda^1\rangle = 0 \\
      (A-\alpha)|\Lambda^0\rangle - (X^3-x^3)|\Lambda^1\rangle = 0
    \end{cases}
  \end{equation}
  By using the $|\alpha\rangle$-representation of the operators, we have
  \begin{equation}
    \begin{cases}
      \int_{\mathbb C} \left[(\varphi_{X^3}(\beta)-x^3)\langle \beta|\Lambda^0\rangle + (\overline{\varphi_A(\beta)}-\bar \alpha)\langle \beta|\Lambda^1\rangle \right] |\beta\rangle \frac{d^2\beta}{\pi} = 0 \\
      \int_{\mathbb C} \left[(\varphi_A(\beta)- \alpha)\langle \beta|\Lambda^0\rangle - (\varphi_{X^3}(\beta)-x^3)\langle \beta|\Lambda^1\rangle \right] |\beta\rangle \frac{d^2\beta}{\pi} = 0
    \end{cases}
  \end{equation}
  $a|\gamma\rangle = \gamma|\gamma\rangle$ it follows that $\int_{\mathbb C} (\beta-\gamma)\langle \beta|\gamma\rangle \frac{d^2\beta}{\pi}=0$ since $\varphi_a(\beta)=\beta$. We conclude that $\exists \gamma,\gamma' \in \mathbb C$, $\exists c,d \in \mathbb C$ such that $\forall \beta \in \mathbb C$ we have:
  \begin{equation}
    \begin{cases}
      (\varphi_{X^3}(\beta)-x^3)\langle \beta|\Lambda^0\rangle + (\overline{\varphi_A(\beta)}-\bar \alpha)\langle \beta|\Lambda^1\rangle  = c (\beta-\gamma) \langle \beta|\gamma \rangle \\
     (\varphi_A(\beta)- \alpha)\langle \beta|\Lambda^0\rangle - (\varphi_{X^3}(\beta)-x^3)\langle \beta|\Lambda^1\rangle = d (\beta-\gamma') \langle \beta|\gamma'\rangle
    \end{cases}
  \end{equation}
  When $\beta=\gamma$ or $\beta=\gamma'$ the first and the second equation become
  \begin{eqnarray}
    \beta=\gamma & \Rightarrow & (\varphi_{X^3}(\gamma)-x^3)\langle \gamma|\Lambda^0\rangle + (\overline{\varphi_A(\gamma)}-\bar \alpha)\langle \gamma|\Lambda^1\rangle  = 0 \\
    \beta=\gamma' & \Rightarrow & (\varphi_A(\gamma')- \alpha)\langle \gamma'|\Lambda^0\rangle - (\varphi_{X^3}(\gamma')-x^3)\langle \gamma'|\Lambda^1\rangle = 0
  \end{eqnarray}
  To have non-trivial solutions, we must have
  \begin{equation}
    \gamma=\gamma'=\varphi_A^{-1}(\alpha) \text{ \& } x^3=\varphi_{X^3}(\gamma) = \varphi_{X^3} \circ \varphi_A^{-1}(\alpha)
  \end{equation}
  It follows that
    \begin{equation}
    \begin{cases}
      (\varphi_{X^3}(\beta)-\varphi_{X^3}(\alpha_A))\langle \beta|\Lambda^0\rangle + (\overline{\varphi_A(\beta)}-\bar \alpha)\langle \beta|\Lambda^1\rangle  = c (\beta-\alpha_A) \langle \beta|\alpha_A \rangle \\
     (\varphi_A(\beta)- \alpha)\langle \beta|\Lambda^0\rangle - (\varphi_{X^3}(\beta)-\varphi_{X^3}(\alpha_A))\langle \beta|\Lambda^1\rangle = d (\beta-\alpha_A) \langle \beta|\alpha_A\rangle
    \end{cases}
    \end{equation}
    \begin{equation}
      \iff \underbrace{\left(\begin{array}{cc} \varphi_{X^3}(\beta)-\varphi_{X^3}(\alpha_A) & \overline{\varphi_A(\beta)}-\bar \alpha \\ \varphi_A(\beta)- \alpha & -(\varphi_{X^3}(\beta)-\varphi_{X^3}(\alpha_A)) \end{array} \right)}_{\slashed{\mathcal D}_x(\beta)} \left(\begin{array}{c} \langle \beta|\Lambda^0\rangle \\ \langle \beta|\Lambda^1\rangle \end{array} \right) = (\beta-\alpha_A) \langle \beta|\alpha_A \rangle \left(\begin{array}{c} c \\ d \end{array} \right)
    \end{equation}
    \begin{equation}
      \iff \left(\begin{array}{c} \langle \beta|\Lambda^0\rangle \\ \langle \beta|\Lambda^1\rangle \end{array} \right) = \frac{(\beta-\alpha_A) \langle \beta|\alpha_A \rangle}{|\det(\slashed{\mathcal D}_x(\beta))|} \left(\begin{array}{cc} \varphi_{X^3}(\beta)-\varphi_{X^3}(\alpha_A) & \overline{\varphi_A(\beta)}-\bar \alpha \\ \varphi_A(\beta)- \alpha & -(\varphi_{X^3}(\beta)-\varphi_{X^3}(\alpha_A)) \end{array} \right)  \left(\begin{array}{c} c \\ d \end{array} \right)
    \end{equation}
    We have then two linearly independent solutions:
    \begin{equation}
      \left(\begin{array}{c} \langle \beta|\Lambda^0\rangle \\ \langle \beta|\Lambda^1\rangle \end{array} \right) = \frac{1}{N_x} \frac{(\beta-\alpha_A) \langle \beta|\alpha_A \rangle}{|\varphi_A(\beta)-\alpha|^2+|\varphi_{X^3}(\beta)-\varphi_{X^3}(\alpha_A)|^2} \left(\begin{array}{c} \overline{\varphi_A(\beta)}-\bar \alpha  \\ -(\varphi_{X^3}(\beta)-\varphi_{X^3}(\alpha_A)) \end{array} \right)
    \end{equation}
    and
    \begin{equation}
      \left(\begin{array}{c} \langle \beta|\Lambda^0_*\rangle \\ \langle \beta|\Lambda^1_*\rangle \end{array} \right) = \frac{1}{N_x} \frac{(\beta-\alpha_A) \langle \beta|\alpha_A \rangle}{|\varphi_A(\beta)-\alpha|^2+|\varphi_{X^3}(\beta)-\varphi_{X^3}(\alpha_A)|^2} \left(\begin{array}{c} \varphi_{X^3}(\beta)-\varphi_{X^3}(\alpha_A) \\ \varphi_A(\beta)- \alpha \end{array} \right)
    \end{equation}
    with a normalisation factor
    \begin{eqnarray}
      N_x^2 & =& \int_{\mathbb C} \llangle \Lambda_\bullet|\beta\rangle\langle \beta|\Lambda_\bullet \rrangle \frac{d^2\beta}{\pi} \\
      & = & \int_{\mathbb C} \frac{|\beta-\alpha_A|^2 |\langle \beta|\alpha_A\rangle|^2}{|\varphi_A(\beta)-\alpha|^2+|\varphi_{X^3}(\beta)-\varphi_{X^3}(\alpha_A)|^2} \frac{d^2\beta}{\pi}
    \end{eqnarray}
    (where $\bullet$ stands for ``$\phantom{*}$'' or ``$*$''). Finally by a change of variable we have
    \begin{eqnarray}
      |\Lambda_\bullet \rrangle & = & \int_{\mathbb C} \left(\begin{array}{c} \langle \beta|\Lambda^0_\bullet \rangle \\ \langle \beta|\Lambda^1_\bullet \rangle \end{array} \right) \otimes |\beta \rangle \frac{d^2\beta}{\pi} \\
      & = & \int_{\mathbb C} \left(\begin{array}{c} \langle \beta+\alpha_A|\Lambda^0_\bullet \rangle \\ \langle \beta+\alpha_A|\Lambda^1_\bullet \rangle \end{array} \right) \otimes |\beta+\alpha_A \rangle \frac{d^2\beta}{\pi}
    \end{eqnarray}
\end{proof}

The previous quasicoherent states are the solutions of the equation $\slashed D_x|\Lambda\rrangle=0$ in $\mathscr F_\infty$. Nothing ensures that they are vectors of $\mathscr F$.\\
This result can be generalised in higher dimension as explained in \ref{highdim}.\\

The entanglement between the spin and the bosonic degrees of freedom on $\mathfrak M$ can be represented by the density matrix of the quasicoherent state:
\begin{eqnarray}
  \rho_\Lambda(x) & = & \tr_{\mathscr F_\infty} |\Lambda(x) \rrangle \llangle \Lambda(x)| \\
  & = & \int_{\mathbb C} \langle \beta|\Lambda(x)\rrangle \llangle \Lambda(x)|\beta\rangle \frac{d^2\beta}{\pi} \\
  & = & \int_{\mathbb C} \frac{|\beta|^2 e^{-|\beta|^2}}{(|\Delta \varphi_A|^2+|\Delta \varphi_{X^3}|^2)^2} \left(\begin{array}{cc} |\Delta \varphi_A|^2 & - \overline{\Delta \varphi_A}\Delta \varphi_{X^3} \\ - \Delta \varphi_A \overline{\Delta \varphi_{X^3}} & |\Delta \varphi_{X^3}|^2 \end{array} \right) \frac{d^2\beta}{\pi N_x^2}
\end{eqnarray}
The R\'enyi entropy $S_\Lambda(x) = - \ln \tr(\rho_\Lambda(x)^2)$ being a measure of the entanglement.\\
Note that $\rho_{\Lambda *}(x) = (\mathrm{co}\rho_\Lambda)^\dagger$ (where $\mathrm{co}X$ denotes the comatrix of $X$).\\

Remark: the convergence of the complex integrals is ensured if at least one of the first derivatives of $\varphi_A$ or $\varphi_{X^3}$ at $\alpha_A = \varphi^{-1}(x^1+\imath x^2)$ is non-zero (and if at least one of the second derivatives of $\varphi_{A}$ or $\varphi_{X^3}$  at $\alpha_A$ is non-zero if its first derivatives are zero), see \ref{convergence}. If it is not the case for some isolated points $x \in M_\Lambda$, we can still extend $\rho_\Lambda$ by continuity.\\

The existence of two linearly independent quasicoherent states can be understood by examining their polarisation local states: $|s(\beta)\rangle = \left(\begin{array}{c} \overline{\Delta \varphi_A} \\ -\Delta \varphi_{X^3} \end{array}\right) \in \mathbb C^2$ and $|s_*(\beta)\rangle = \left(\begin{array}{c} \Delta \varphi_{X^3} \\ \Delta \varphi_A \end{array}\right) \in \mathbb C^2$. $\langle s_*|s\rangle=0$ and then $\langle s_*(\beta)|\vec \sigma|s_*(\beta)\rangle = - \langle s(\beta)|\vec \sigma|s(\beta)\rangle$, the exchange $\circ \leftrightarrow *$ is then a local (at $\beta$) reversal of orientation of $\mathfrak M$.  The reversal of orientation is local in the $\beta$-representation of $(|\Lambda\rrangle,|\Lambda_*\rrangle)$; in the $x$-representation, since the states are not points, the reversal of orientation is non-local (the state at $x$ depends on the integration over all $\beta$). As in commutative geometry, at the neighbourhood of a point on a surface, it exists two choices of orientation. We have the same things with $\mathfrak M$ but with a non-local meaning (integration over $\beta$), possible quantum superposition of orientation states and entanglement between local orientation states and the bosonic degrees of freedom. In M-theory, since $\slashed D_x$ is the Dirac operator for a massless fermionic string, the exchange can be interpreted as $\vec \sigma \leftrightarrow -\vec \sigma$ and then as the exchange between left and right fermionic sectors in the Weyl representation of the Dirac operator.

\subsection{Example: the noncommutative plane}
The noncommutative plane is defined by $X^1 = L \frac{a+a^+}{2}$, $X^2 = L \frac{a-a^+}{2\imath}$ and $X^3=0$, $L \in \mathbb R^{+*}$ being a scale parameter. In other words:
\begin{equation}
  \slashed D_x = \left(\begin{array}{cc} x^3 & La^+-\bar \alpha \\ La-\alpha & -x^3 \end{array} \right)
\end{equation}
$[X^1,X^2]=\frac{\imath L^2}{2}$ and $A = La$. We have then $\varphi_A(\beta) = L\beta$ and $\varphi_{X^3}(\beta)=0$. $M_\Lambda$ is then the plane $x^3=0$ and the quasicoherent states are
\begin{eqnarray}
  |\Lambda(x)\rrangle & = & \left(\begin{array}{c} 1 \\ 0 \end{array} \right) \otimes \int_{\mathbb C} \frac{\beta \langle \beta+\alpha/L|\alpha/L\rangle}{L\beta} |\beta+\alpha/L\rangle \frac{d^2\beta}{\pi/L} \\
  & = & \left(\begin{array}{c} 1 \\ 0 \end{array} \right) \otimes |\alpha/L\rangle
\end{eqnarray}
with $\alpha = x^1+\imath x^2$; and
\begin{eqnarray}
  |\Lambda_*(x)\rrangle & = & \left(\begin{array}{c} 0 \\ 1 \end{array} \right) \otimes \int_{\mathbb C} \frac{\beta \langle \beta+\alpha/L |\alpha/L\rangle}{L\bar \beta} |\beta+\alpha/L\rangle \frac{d^2\beta}{\pi/L} \\
  & = & \left(\begin{array}{c} 0 \\ 1 \end{array} \right) \otimes |\alpha/L\rangle_*
\end{eqnarray}
with $|\alpha\rangle_* = \int_{\mathbb C} e^{2\imath \arg (\beta-\alpha)} \langle \beta|\alpha\rangle |\beta \rangle \frac{d^2\beta}{\pi}$.\\
We have well $|\Lambda \rrangle \in \mathbb C^2 \otimes \mathscr F$. In contrast $|\Lambda_*\rrangle \not\in \mathbb C^2 \otimes \mathscr F$ since by construction $|\alpha\rangle_*$ is eigenvector of $a^+$ : $a^+|\alpha\rangle_* = \bar \alpha |\alpha\rangle_*$, whereas $a^+$ has no eigenvector in $\mathscr F$. We can see it with the Fock-Bargman representation \cite{Perelomov} where $a \to \frac{d}{dz}$, $a^+ \to z$ and $|n \rangle \to \frac{z^n}{\sqrt{n!}}$. The equation $a|\alpha \rangle = \alpha |\alpha \rangle$ becomes $\frac{d}{dz} u_\alpha(z) = \alpha u_\alpha(z)$ and has for solution $u_\alpha(z) = e^{-|\alpha|^2/2} e^{\alpha z}$ whereas the equation $a^+|\alpha\rangle_* = \bar \alpha |\alpha\rangle_*$ becomes $zv_\alpha(z)=\bar \alpha v_\alpha(z)$ which has for solution $v_\alpha(z) = \delta(z-\bar \alpha)$. It is clear that $|\alpha \rangle_*$ is a singular distribution in $\mathscr F_\infty$ and not a regular state of $\mathscr F$.\\

The magnetic potential living on $\mathfrak M$ is then
\begin{eqnarray}
  \mathrm A & = & -\imath \llangle \Lambda|d|\Lambda \rrangle \\
  & = & -\imath \left(\langle \alpha/L| \frac{\partial}{\partial \alpha} |\alpha/L\rangle d\alpha + \langle \alpha/L| \frac{\partial}{\partial \bar \alpha} |\alpha/L\rangle d\bar \alpha\right) \\
  & = & \frac{-\imath}{L} \frac{\bar \alpha d\alpha-\alpha d\bar \alpha}{2} \\
  & = & \frac{x^1dx^2-x^2dx^1}{L} \\
  & = & \frac{1}{L} r^2d\theta
\end{eqnarray}
with $\alpha = re^{\imath \theta}$. It follows that $\mathrm B = d\mathrm A = \frac{2}{L} rdr \wedge d\theta$. It does not appear as a magnetic field induced by magnetic monopoles. The situation is as if the plane $M_\Lambda$ were inside an infinite solenoid of infinite radius with $\frac{\nu I}{2} = \frac{1}{L}$. The adiabatic transport of $|\Lambda \rrangle$ along a closed path $\mathscr C$ on $M_\Lambda$ is then
\begin{eqnarray}
  |\Psi(T)\rrangle & = & e^{-\imath \oint_{\mathscr C} \mathrm A} |\Lambda(x_0)\rrangle \\
  & = & e^{-\imath \frac{2}{L} \mathcal A_{\mathscr S}} \left(\begin{array}{c} 1 \\ 0 \end{array} \right) \otimes |\alpha_0/L\rangle
\end{eqnarray}
where $\mathcal A_{\mathscr S}$ is the area of the surface in $M_\Lambda$ delimited by $\mathscr C$. Note that the adiabatic dynamics induced by $\slashed D_x \in \mathcal L(\mathbb C^2 \otimes \mathscr F)$ does not couple $|\Lambda \rrangle \in \mathbb C^2 \otimes \mathscr F$ and $|\Lambda_* \rrangle \not\in \mathbb C^2 \otimes \mathscr F$.

\subsection{Magnetic potential onto a CCR D2-brane}
\begin{prop}
  Let $|\Lambda\rrangle$ and $|\Lambda_*\rrangle$ be the quasicoherent states of a CCR D2-brane, and $\alpha_A = \varphi_A^{-1}(\alpha)$. The magnetic potential associated with the adiabatic transport is
  \begin{eqnarray}
    \mathrm A & = & -\imath \llangle \Lambda|d|\Lambda \rrangle \\
    & = & \frac{-\imath}{2}((\bar \alpha_A+\bar \delta_A) d\alpha_A - (\alpha_A+\delta_A) d\bar \alpha_A) \nonumber \\
    & & + \int_{\mathbb C} \mathscr A(\beta) |\beta|^2 e^{-|\beta|^2} \frac{d^2\beta}{\pi} \nonumber \\
    & & + \imath d\ln N_x
  \end{eqnarray}
  with
  \begin{equation}
    \delta_A = 2\int_{\mathbb C} \frac{\beta|\beta|^2e^{-|\beta|^2}}{|\Delta \varphi_A(\beta)|^2+|\Delta \varphi_{X^3}(\beta)|^2} \frac{d^2\beta}{\pi N_x^2}
  \end{equation}
  and
  \begin{equation}
    \mathscr A(\beta) = \imath \frac{\overline{\Delta \varphi_A(\beta)}d\Delta \varphi_A(\beta)+\Delta \varphi_{X^3}(\beta)d\Delta \varphi_{X^3}(\beta)}{(|\Delta \varphi_A(\beta)|^2+|\Delta \varphi_{X^3}(\beta)|^2)^2 N_x^2}
  \end{equation}
  The expression of $\mathrm A_* = -\imath \llangle \Lambda_*|d|\Lambda_* \rrangle$ is similar except that $\mathscr A$ is replaced by $-\overline{\mathscr A}$.\\
  The adiabatic coupling between $|\Lambda\rrangle$ and $|\Lambda_*\rrangle$ is
  \begin{eqnarray}
    \mathrm C & = & -\imath \llangle \Lambda_*|d|\Lambda \rrangle \\
    & = & -\imath \int_{\mathbb C} \frac{\Delta \varphi_{X^3}(\beta)d\overline{\Delta \varphi_A(\beta)}-\Delta \varphi_{A}(\beta)d\Delta \varphi_{X^3}(\beta)}{(|\Delta \varphi_A(\beta)|^2+|\Delta \varphi_{X^3}(\beta)|^2)^2}  |\beta|^2 e^{-|\beta|^2} \frac{d^2\beta}{\pi N_x^2}
  \end{eqnarray}
\end{prop}

$ \frac{-\imath}{2}(\bar \alpha_A d\alpha_A - \alpha_A d\bar \alpha_A) = [\varphi_A^{-1}]^* \mathrm A_{NCP}$ is the image of the noncommutative plane potential $\mathrm A_{NCP}$ of the previous subsection by the push-forward of $\varphi_A^{-1}$. $\frac{-\imath}{2} (\bar \delta_A d\alpha_A - \delta_A d\bar \alpha_A)$ is the effect of the deformation of the noncommutative plane, $\delta_A$ is zero if $|\Delta \varphi_A|^2$ and $|\Delta \varphi_{X^3}|^2$ are even functions of $\beta$. $\int \mathscr A |\beta|^2 e^{-|\beta|^2} \frac{d^2\beta}{\pi}$ is the potential induced by the embedding of the deformed noncommutative plane in $\mathbb R^3$. $\imath d\ln N_x$ is just a gauge change restoring the real valued property of $\mathrm A$.

\begin{proof}
  We write the quasicoherent state as
  \begin{equation}
    |\Lambda \rrangle = \int_{\mathbb C} \beta e^{-|\beta|^2/2} e^{\imath \IM(\bar \beta \alpha_A)} |\xi_x(\beta)\rangle \otimes |\beta+\alpha_A\rangle \frac{d^2\beta}{\pi}
  \end{equation}
  with
  \begin{equation}
    |\xi_x \rangle = \frac{1}{N_x} \frac{1}{|\Delta \varphi_A|^2+|\Delta \varphi_{X^3}|^2} \left(\begin{array}{c} \overline{\Delta \varphi_A} \\ - \Delta_{\varphi_{X^3}} \end{array} \right)
  \end{equation}
  We have then
  \begin{eqnarray}
    d|\Lambda \rrangle & = & \frac{1}{2} \int_{\mathbb C} (\bar \beta d\alpha_A-\beta d\bar \alpha_A) \beta e^{-|\beta|^2/2} e^{\imath \IM(\bar \beta \alpha_A)} |\xi_x(\beta)\rangle \otimes |\beta+\alpha_A\rangle \frac{d^2\beta}{\pi} \nonumber \\
    & & + \int_{\mathbb C} \beta e^{-|\beta|^2/2} e^{\imath \IM(\bar \beta \alpha_A)} d|\xi_x(\beta)\rangle \otimes |\beta+\alpha_A\rangle \frac{d^2\beta}{\pi} \nonumber \\
    & & + \int_{\mathbb C} \beta e^{-|\beta|^2/2} e^{\imath \IM(\bar \beta \alpha_A)} |\xi_x(\beta)\rangle \otimes d|\beta+\alpha_A\rangle \frac{d^2\beta}{\pi}
  \end{eqnarray}
  By using the expression of $|\alpha\rangle$ in the canonical basis of $\mathscr F$ we have
  \begin{equation}
    d|\alpha\rangle = -\frac{\alpha d\bar \alpha + \bar \alpha d\alpha}{2}|\alpha\rangle + a^+|\alpha\rangle d\alpha
  \end{equation}
  It follows that
  \begin{eqnarray}
    & & \int_{\mathbb C} \beta e^{-|\beta|^2/2} e^{\imath \IM(\bar \beta \alpha_A)} |\xi_x(\beta)\rangle \otimes d|\beta+\alpha_A\rangle \frac{d^2\beta}{\pi} \nonumber \\
    & & = -\frac{1}{2} \int_{\mathbb C} ((\beta+\alpha_A)d\bar \alpha_A+(\bar \beta+\bar \alpha_A)d\alpha_A) \beta e^{-|\beta|^2/2} e^{\imath \IM(\bar \beta \alpha_A)} |\xi_x(\beta)\rangle \otimes |\beta+\alpha_A\rangle \frac{d^2\beta}{\pi} \nonumber \\
    & & + \int_{\mathbb C} \beta e^{-|\beta|^2/2} e^{\imath \IM(\bar \beta \alpha_A)} |\xi_x(\beta)\rangle \otimes a^+|\beta+\alpha_A\rangle \frac{d^2\beta}{\pi}
  \end{eqnarray}
  and so
  \begin{eqnarray}
    d|\Lambda \rrangle & = & -\frac{1}{2} (\alpha_Ad\bar \alpha_A+\bar \alpha_Ad\alpha_A) \int_{\mathbb C} \beta e^{-|\beta|^2/2} e^{\imath \IM(\bar \beta \alpha_A)} |\xi_x(\beta)\rangle \otimes |\beta+\alpha_A\rangle \frac{d^2\beta}{\pi} \nonumber \\
    & & - \int_{\mathbb C} \beta^2 e^{-|\beta|^2/2} e^{\imath \IM(\bar \beta \alpha_A)} |\xi_x(\beta)\rangle \otimes |\beta+\alpha_A\rangle \frac{d^2\beta}{\pi} d \bar \alpha_A \nonumber \\
    & & + \int_{\mathbb C} \beta e^{-|\beta|^2/2} e^{\imath \IM(\bar \beta \alpha_A)} d|\xi_x(\beta)\rangle \otimes |\beta+\alpha_A\rangle \frac{d^2\beta}{\pi} \nonumber \\
    & & + \int_{\mathbb C} \beta e^{-|\beta|^2/2} e^{\imath \IM(\bar \beta \alpha_A)} |\xi_x(\beta)\rangle \otimes a^+|\beta+\alpha_A\rangle \frac{d^2\beta}{\pi} d\alpha_A
  \end{eqnarray}
  It follows that
  \begin{eqnarray}
    \llangle \Lambda|d|\Lambda \rrangle & = & -\frac{1}{2} (\alpha_Ad\bar \alpha_A+\bar \alpha_Ad\alpha_A) \nonumber \\
    & & - \int_{\mathbb C} \beta|\beta|^2 e^{-|\beta|^2} \langle \xi_x |\xi_x\rangle \frac{d^2\beta}{\pi} d \bar \alpha_A \nonumber \\
    & & +\int_{\mathbb C} |\beta|^2 e^{-|\beta|^2} \langle \xi_x |d|\xi_x\rangle \frac{d^2\beta}{\pi} \nonumber \\
    & & +\int_{\mathbb C} |\beta|^2 e^{-|\beta|^2} \langle \xi_x |\xi_x\rangle (\bar \beta+\bar \alpha_A) \frac{d^2\beta}{\pi} d\alpha_A \\
    & = & \frac{1}{2} (\bar \alpha_Ad\alpha_A- \alpha_Ad\bar \alpha_A) \nonumber \\
    & & + \int_{\mathbb C} (\bar \beta d\alpha_A-\beta d\bar \alpha_A) |\beta|^2 e^{-|\beta|^2} \langle \xi_x |\xi_x\rangle \frac{d^2\beta}{\pi} \nonumber \\
    & & + \int_{\mathbb C} |\beta|^2 e^{-|\beta|^2} \langle \xi_x |d|\xi_x\rangle \frac{d^2\beta}{\pi}
  \end{eqnarray}
  Finally
  \begin{eqnarray}
    \langle \xi_x|d|\xi_x\rangle & = & -\frac{dN_x}{N_x^3} \frac{1}{|\Delta \varphi_A|^2+|\Delta \varphi_{X^3}|^2} \nonumber \\
    & & - \frac{1}{N_x^2} \frac{d|\Delta \varphi_A|^2+d|\Delta \varphi_{X^3}|^2}{(|\Delta \varphi_A|^2+|\Delta \varphi_{X^3}|^2)^2} \nonumber \\
    & & + \frac{1}{N_x^2} \frac{\Delta \varphi_A d\overline{\Delta \varphi_A}+\Delta \varphi_{X^3} d\Delta \varphi_{X^3}}{(|\Delta \varphi_A|^2+|\Delta \varphi_{X^3}|^2)^2} \\
    & = & -\frac{dN_x}{N_x^3} \frac{1}{|\Delta \varphi_A|^2+|\Delta \varphi_{X^3}|^2} \nonumber \\
    & & -  \frac{1}{N_x^2} \frac{\overline{\Delta \varphi_A} d\Delta \varphi_A +\Delta \varphi_{X^3} d\Delta \varphi_{X^3}}{(|\Delta \varphi_A|^2+|\Delta \varphi_{X^3}|^2)^2}
  \end{eqnarray}
  and then
  \begin{eqnarray}
    & & \int_{\mathbb C} |\beta|^2 e^{-|\beta|^2} \langle \xi_x |d|\xi_x\rangle \frac{d^2\beta}{\pi} \nonumber \\
    & = & -\frac{dN_x}{N_x^3}\underbrace{\int_{\mathbb C} \frac{1}{|\Delta \varphi_A|^2+|\Delta \varphi_{X^3}|^2} |\beta|^2 e^{-|\beta|^2} \frac{d^2\beta}{\pi}}_{=N_x^2}  \nonumber \\
    & & -  \frac{1}{N_x^2} \int_{\mathbb C} \frac{\overline{\Delta \varphi_A} d\Delta \varphi_A +\Delta \varphi_{X^3} d\Delta \varphi_{X^3}}{(|\Delta \varphi_A|^2+|\Delta \varphi_{X^3}|^2)^2}|\beta|^2 e^{-|\beta|^2} \frac{d^2\beta}{\pi}
  \end{eqnarray}
\end{proof}

Note that the convergence of the complex integrals in the formulae of $\mathrm A$ and $\mathrm C$ is ensured in the same conditions that for $\rho_\Lambda$, see \ref{convergence}.

\section{Noncommutative cylinder and Aharonov-Bohm effect}
We consider a CCR D2-brane wrapped around an axis to form a noncommutative cylinder.

\subsection{The noncommutative cylinder}
The topological cylinder is defined by the following quotient of the plane: $\mathcal C^2=\mathbb R \times \mathbb R/(2\pi\ell \mathbb Z)$, with $\ell \in \mathbb R^{+*}$ the wrap parameter. Let $(u,\theta)$ be coordinates onto $\mathcal C^2$ ($\theta= \frac{u'}{\ell} \mod 2\pi$ where $(u,u')$ are the coordinates in the plane). Analytical functions on $\mathcal C^2$ are then such that $g(u,\theta) = \sum_{k=0}^\infty \sum_{n\in \mathbb Z} g_{kn} u^k e^{\imath n \theta}$ (with $g_{kn} \in \mathbb C$). The standard embedding of the cylinder into the space $f:\mathcal C^2 \to \mathbb R^3$ is defined by $f(u,\theta) = (R\cos \theta,R\sin \theta,Lu)$ with $R \in \mathbb R^{+*}$ the cylinder radius and $L \in \mathbb R^{+*}$ a scale parameter.\\
The topological noncommutative cylinder is then the algebra $\mathfrak X_{\mathcal C^2}$ generated by $U = \frac{a+a^+}{2}$ (quantisation of $u$) and $\Theta = e^{\frac{a-a^+}{2\ell}}$ (quantisation of $e^{\imath \theta}$). Note that $\Theta = D(-\frac{1}{2\ell})$ where $D(\alpha) = e^{\alpha a^+-\bar \alpha a}$ is the displacement operator \cite{Perelomov} generating the coherent states: $D(\alpha)|0\rangle = |\alpha\rangle$. We can note that $[U,\Theta]=-\frac{1}{2\ell} \Theta$. By considering the embedding $f$ of the classical case, the Dirac operator of the noncommutative cylinder is
\begin{equation}
  \slashed D_x = \left(\begin{array}{cc} LU-x^3 & R\Theta^\dagger-\bar \alpha \\ R\Theta-\alpha & -(LU-x^3) \end{array} \right)
\end{equation}
with $X^3=L\frac{a+a^+}{2}$ and $A=Re^{\frac{a-a^+}{2\ell}}$. Because $\varphi_{D(\alpha)}(\beta) = e^{|\alpha|^2/2} e^{\alpha\bar \beta-\bar \alpha \beta}$ (see \cite{Perelomov} and \ref{alphaRep}) we have:
\begin{equation}
  \varphi_{X^3}(\beta) = L \RE(\beta) \qquad \varphi_A(\beta) = Re^{\frac{1}{8\ell^2}} e^{\frac{\imath}{\ell} \IM(\beta)}
\end{equation}
We denote by $R_\ell \equiv Re^{\frac{1}{8\ell^2}}$ the corrected radius. The eigensurface $M_\Lambda$ is then described with
\begin{eqnarray}
  & & \alpha \in \Ran\varphi_A = \{R_\ell e^{\imath \theta}\}_{\theta \in [0,2\pi[} \\
  & \Rightarrow & \alpha_{An} = \varphi_A^{-1}(R_\ell e^{\imath \theta}) = u + \imath \ell(\theta+2n\pi)\qquad u\in\mathbb R, n\in \mathbb Z \\
  & & x^3 = \varphi_{X^3}(\alpha_A) = Lu
\end{eqnarray}
In other words, $x=(R_\ell \cos\theta,R_\ell \sin \theta,Lu)$. The quasicoherent states of the noncommutative cylinder are then
\begin{eqnarray}
  |\Lambda_n(x)\rrangle & = & \int_{\mathbb C} \lambda_n(\beta) \left(\begin{array}{c} R_\ell(e^{-\imath \frac{\IM(\beta)}{\ell}}-1)e^{-\imath \theta} \\  -L\RE(\beta) \end{array}\right)\otimes |\beta+\alpha_{An}\rangle \frac{d^2\beta}{\pi N}\\
  |\Lambda_{*n}(x)\rrangle & = & \int_{\mathbb C} \lambda_n(\beta) \left(\begin{array}{c} L\RE(\beta) \\ R_\ell(e^{\imath \frac{\IM(\beta)}{\ell}}-1)e^{\imath \theta} \end{array}\right)\otimes |\beta+\alpha_{An}\rangle \frac{d^2\beta}{\pi N}
\end{eqnarray}
with
\begin{equation}
  \lambda_n(\beta)= \frac{\beta e^{-\frac{|\beta|^2}{2}}e^{-\imath u\IM(\beta)} e^{\imath\ell \RE(\beta)(\theta+2n\pi)}}{R_\ell^2|e^{\imath\frac{\IM(\beta)}{\ell}}-1|^2+L^2|\RE(\beta)|^2}
\end{equation}
the normalisation factor being
\begin{equation}
  N^2 = \int_{\mathbb C} \frac{e^{-|\beta|^2}|\beta|^2}{R_\ell^2|e^{\imath \frac{\IM(\beta)}{\ell}}-1|^2+L^2|\RE(\beta)|^2} \frac{d^2\beta}{\pi}
\end{equation}
The index $n$ is associated with the periodicity:
\begin{equation}
  |\Lambda_n(u,\theta+2p\pi)\rrangle = |\Lambda_{n+p}(u,\theta)\rrangle
\end{equation}
which is just a phase change in $|\beta\rangle$-representation:
\begin{equation}
  \langle \beta|\Lambda_{n+p}(x)\rrangle = e^{2\imath p \pi \ell \RE(\beta)} \langle \beta|\Lambda_n(x)\rrangle
\end{equation}

\subsection{Magnetic field of the noncommutative cylinder}
The magnetic potential onto the noncommutative cylinder is (see \ref{CompCylind}):
\begin{equation}
  \mathrm A = -\kappa d\theta + \ell(ud\theta-(\theta+2n\pi)du)
\end{equation}
where $\kappa = \int_{\mathbb C} \frac{|\beta|^2 e^{-|\beta|^2} R_\ell^2|e^{\imath\frac{\IM(\beta)}{\ell}}-1|^2}{(R_\ell^2|e^{\imath\frac{\IM(\beta)}{\ell}}-1|^2+L^2|\RE(\beta)|^2)^2} \frac{d^2\beta}{\pi N^2}$ is a positive constant.
\begin{itemize}
\item $\mathrm A_{geo} = \ell(ud\theta-(\theta+2n\pi)du) = \ell(udu'-u'du)$ is the magnetic potential of the noncommutative plane from which the cylinder comes. Note that $\ell2n\pi du$ is just a gauge change and we can redefine the magnetic potential as $\mathrm A_{geo} = \ell(ud\theta-\theta du)$. The magnetic field is $\mathrm B_{geo} = d\mathrm A_{geo} = 2\ell du \wedge d\theta$ ($\vec {\mathrm B}_{geo} = - \frac{2\ell}{RL} \vec e_r$). This is the magnetic field produced by a magnetically charged infinite line (the $x^3$-axis) with the magnetic monopole linear density $-\frac{4\pi \ell}{L}$.
\item $\mathrm A_{topo} = -\kappa d\theta$ is the magnetic potential produced by an infinite solenoid along the $x^3$-axis inside the noncommutative cylinder (the radius of the solenoid is lower than $R$ and can be identified to $0$).
\end{itemize}
The cylinder axis appears then as magnetically charged line surrounded by a solenoid. The adiabatic transport of $|\Lambda_n \rrangle$ along a closed path $\mathscr C$ on the eigencylinder $M_\Lambda$ is then
\begin{equation}
  |\Psi_n(T) \rrangle = e^{\imath 2\pi \kappa p} e^{-\imath \oint_{\mathscr C} \mathrm A_{geo}} |\Lambda_{n+p}(x_0)\rrangle
\end{equation}
where $p$ is the number of turns of $\mathscr C$ around the $x^3$-axis. The topological phase $e^{\imath 2\pi \kappa p}$ is a kind of Aharonov-Bohm effect. (Note that the adiabatic couplings are zero: $\llangle \Lambda_*|d|\Lambda\rrangle=0$, so an adiabatic transport starting on $|\Lambda\rrangle$ does not involve $|\Lambda_*\rrangle$).\\

With $\ell \gg 1$, by writing $e^{\imath \IM(\beta)/\ell} -1 \simeq \imath \IM(\beta)/\ell$ in the expression of $\kappa$, we have $\kappa \simeq \frac{1}{2}$. Moreover $\lim_{\ell \to 0} \kappa = 0$. Numerical evaluation of $\kappa$ can be found fig. \ref{kappaCyl}.
\begin{figure}
  \begin{center}
    \includegraphics[width=8cm]{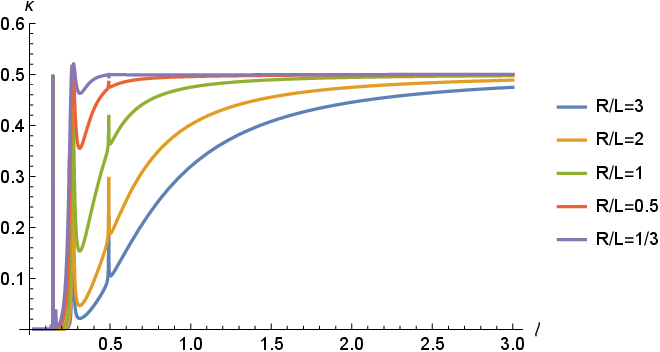}
    \end{center}
  \caption{\label{kappaCyl} The topological index of the noncommutative cylinder (solenoid current) with respect to the wrap parameter $\ell$ for different values of $\frac{R}{L}$ (ratio between the cylinder radius and the scale factor).}
\end{figure}

The quasicoherent density matrix is:
\begin{equation}
  \rho_\Lambda = \left(\begin{array}{cc} \kappa & 0 \\ 0 & 1-\kappa \end{array} \right)
\end{equation}
The topological index $\kappa$ is then the population of the spin state $|\uparrow\rangle = {1 \choose 0}$, and the entanglement entropy is $S_\Lambda = - \ln(\kappa^2+(1-\kappa)^2)$ with $\lim_{\ell \to 0} S_\Lambda =0$ and $\lim_{\ell \to +\infty} S_\Lambda = \ln 2$ (maximal entanglement).\\

Note that $\mathrm A_{geo,*} = \mathrm A_{geo}$ and $\mathrm A_{topo,*} = -\mathrm A_{topo}$ (the state $|\Lambda_*\rrangle$ corresponds to reverse the solenoid current). With $|\Psi_n(0)\rrangle = \frac{1}{\sqrt 2}(|\Lambda_n(x_0)\rrangle+|\Lambda_{*n}(x_0)\rrangle)$ we have then
\begin{equation}
  |\Psi_n(T)\rrangle = \frac{e^{-\imath \oint_{\mathscr C} \mathrm A_{geo}}}{\sqrt 2} (e^{\imath 2\pi \kappa p} |\Lambda_{n+p}(x_0)\rrangle + e^{-\imath 2\pi \kappa p} |\Lambda_{*,n+p}(x_0)\rrangle
\end{equation}
where $p$ is still the number of turns of $\mathscr C$ around the $x^3$-axis. It follows that the survival probability is
\begin{equation}
  |\llangle \Psi_{n+p}(0)|\Psi_n(T) \rrangle|^2 = \cos^2(2\pi \kappa p)
\end{equation}
The geometric phase being cancelled, we have a measurable purely topological effect as in usual Aharonov-Bohm effect.

\section{Quantum M\"obius strip}
\subsection{The noncommutative M\"obius strip}
The classical M\"obius strip can be topologically defined as the cylinder by $\mathcal M^2 = \mathbb R \times \mathbb R/(2\pi\ell \mathbb Z)$, but with an algebra of analytical functions satisfying $g(u,\theta+2\pi) = g(-u,\theta)$. It follows that $g(u,\theta) = \sum_{k=0}^{+\infty} \sum_{n \in \mathbb Z} g_{kn} (e^{\imath \frac{\theta}{2}} u)^k e^{\imath n \theta}$ ($g_{kn} \in \mathbb C$). The standard embedding $f:\mathcal M^2 \to \mathbb R^3$ is defined $f(u,\theta) = ((R+Lu\cos \frac{\theta}{2})\cos \theta, (R+Lu\cos(\frac{\theta}{2}))\sin \theta, Lu \sin \frac{\theta}{2})$ with $R \in \mathbb R^{+*}$ the radius and $L \in \mathbb R^{+*}$ a scale parameter.\\
The topological noncommutative M\"obius strip should be defined by an algebra $\mathfrak X_{\mathcal M^2}$ generated by two operators $U$ and $\Theta$. The quantisation of $e^{\imath \theta}$ is simply defined by $\Theta = e^{\frac{a-a^+}{2\ell}}$. But $U$ is submitted to an ambiguity of quantisation: $e^{\frac{a-a^+}{4\ell}} \frac{a+a^+}{2}$, $\frac{a+a^+}{2} e^{\frac{a-a^+}{4\ell}}$ or $e^{\frac{a-a^+}{8\ell}} \frac{a+a^+}{2} e^{\frac{a-a^+}{8\ell}}$? To solve this problem, it seems more efficient to directly define $U$ in the $|\alpha\rangle$-representation:
\begin{equation}
  U = \int_{\mathbb C} e^{\IM(\beta)/(2\ell)} \RE(\beta) |\beta\rangle \langle \beta| \frac{d^2\beta}{\pi}
\end{equation}
or equivalently, to define $A$ and $X^3$ by:
\begin{eqnarray}
  \varphi_A(\beta) & = & (R + L \RE(\beta) \cos(\IM(\beta)/(2\ell))) e^{\imath \frac{\IM(\beta)}{\ell}} \\
  \varphi_{X^3}(\beta) & = & L \RE(\beta) \sin(\IM(\beta)/(2\ell))
\end{eqnarray}
Since $\varphi_A(\beta) = {_*}\langle \beta|;A;|\beta\rangle_*$ where $|\beta\rangle_*$ is the singular distribution eigenvector of $a^+$ and $;A;$ is the anti-normal order of $A$ (see \ref{alphaRep}), we have
\begin{eqnarray}
  A & = & Re^{\frac{a}{2\ell}}e^{-\frac{a^+}{2\ell}} + \frac{L}{4}ae^{\frac{3a}{4\ell}}e^{-\frac{3a^+}{4\ell}} + \frac{L}{4}ae^{\frac{a}{4\ell}}e^{-\frac{a^+}{4\ell}} + \frac{L}{4}e^{\frac{3a}{4\ell}}e^{-\frac{3a^+}{4\ell}}a^+ + \frac{L}{4}e^{\frac{a}{4\ell}}e^{-\frac{a^+}{4\ell}} a^+ \\
  & = & \left(R_\ell + L_\ell' \frac{a+a^++\frac{3}{4\ell}}{2} \frac{e^{\frac{a-a^+}{4\ell}}}{2} + L_\ell \frac{a+a^++\frac{1}{4\ell}}{2} \frac{e^{-\frac{a-a^+}{4\ell}}}{2}\right) e^{\frac{a-a^+}{2\ell}}
\end{eqnarray}
and
\begin{eqnarray}
  X^3 & = & \frac{L}{4\imath}\left(ae^{\frac{a}{4\ell}}e^{-\frac{a^+}{4\ell}} - ae^{-\frac{a}{4\ell}}e^{\frac{a^+}{4\ell}} + e^{\frac{a}{4\ell}}e^{-\frac{a^+}{4\ell}}a^+ - e^{-\frac{a}{4\ell}}e^{\frac{a^+}{4\ell}}a^+ \right) \\
  & = & L_\ell \left(\frac{a+a^+}{2} \frac{e^{\frac{a-a^+}{4\ell}}-e^{-\frac{a-a^+}{4\ell}}}{2\imath} + \frac{1}{8\ell} \frac{e^{\frac{a-a^+}{4\ell}}+e^{-\frac{a-a^+}{4\ell}}}{2\imath}  \right)
\end{eqnarray}
with $R_\ell = Re^{- \frac{1}{8\ell^2}}$, $L_\ell = Le^{-\frac{1}{32\ell^2}}$ and $L'_\ell = Le^{-\frac{9}{32\ell^2}}$.\\

The eigensurface $M_\Lambda$ is then described with
\begin{eqnarray}
  & & \alpha \in \Ran \varphi_A = \{(R + L u \cos \frac{\theta}{2})e^{\imath \theta}\}_{u \in \mathbb R,\theta \in [0,2\pi[} \\
  & \Rightarrow & \alpha_{An} = u + \imath \ell (\theta+2n\pi) \qquad n \in \mathbb Z \\
  & & x^3 = L u \sin \frac{\theta}{2}
\end{eqnarray}

The computation of the quasicoherent state is not obvious in this case, we consider the approximation of this one when the wrap parameter is large $\ell \gg 1$. We suppose moreover than $L = \frac{\mathring L}{\ell}$.\\

The quasicoherent density matrix is then (see \ref{rho_mobius})
\begin{eqnarray}
  \rho_\Lambda(x) & = & \frac{1}{2} \left(\begin{array}{cc} 1+\cos^2 \frac{\theta}{2} & - \frac{e^{-\imath \theta}}{2} \sin \theta \\ - \frac{e^{\imath \theta}}{2} \sin \theta & \sin^2 \frac{\theta}{2} \end{array}\right) \nonumber \\
  & & - \frac{\imath}{4\ell} \frac{\mathring L}{R} u \cos \frac{\theta}{2} \left(\begin{array}{cc} 0 & -e^{-\imath \theta} \\ e^{\imath \theta} & 0 \end{array} \right) \nonumber \\
  & & + \mathcal O\left(\frac{1}{\ell^2}\right)
\end{eqnarray}
We can write $\rho_\Lambda(u,\theta) = \rho_{\Lambda,0}(\theta) + \frac{1}{\ell} \tau(u,\theta) + \mathcal O(\frac{1}{\ell^2})$ with $\rho_{\Lambda,0}(\theta+2p\pi) = \rho_{\Lambda,0}(\theta)$ and $\tau(u,\theta+2p\pi) = \tau((-1)^pu,\theta) = (-1)^p\tau(u,\theta) $. Let $\mathscr C$ be a closed path drawn on the eigensurface $M_\Lambda$, the transport of the spin mixed state $\rho_\Lambda(x(s)) = \tr_{\mathscr F_\infty} |\Lambda(x(s))\rrangle \llangle \Lambda(x(s))|$ (with $s \in [0,1]$ a curvilinear coordinates onto $\mathscr C$) is then
\begin{equation}
  \rho_\Lambda(x(1)) = \rho_{\Lambda,0}(\theta_0) + (-1)^p \frac{1}{\ell} \tau(u_0,\theta_0) + \mathcal O\left(\frac{1}{\ell^2}\right)
\end{equation}
where $p$ is the number of turns of $\mathscr C$ around $0$. The noncommutative M\"obius strip is non-orientable in the sense that the transport of its spin state around a closed path returns to itself if and only if the number of turns is even.  Note that the associated R\'enyi entropy is $S_\Lambda(u,\theta) = - \ln \frac{3+\cos \theta}{4} + \mathcal O\left(\frac{1}{\ell^2}\right)$. The entanglement is then $0$ at $\theta=0$ (where the M\"obius strip identifies with the horizontal plane) and is maximal ($\ln 2$) at $\theta=\pi$ (where the surface of the M\"obius strip is almost vertical).

\subsection{Magnetic potential of the noncommutative M\"obius strip}
The magnetic potential of noncommutative M\"obius strip is $\mathrm A = \mathrm A_{geo} + \mathrm A_{def} + \mathrm A_{topo}$ with the following potential generating a geometric phase issuing from the noncommutative plane:
\begin{equation}
  \mathrm A_{geo}  =  \ell(ud\theta-(\theta+2n\pi)du)
\end{equation}
the following potential associated with the deformation of the plane (see \ref{Adef_mobius}):
\begin{equation}
  \mathrm A_{def} = -\frac{\mathring LR}{(\mathring L+R)^2} \cos \frac{\theta}{2} d\theta + \mathcal O\left(\frac{1}{\ell}\right)
\end{equation}
and the following potential generating a topological phase (see \ref{Atopo_mobius}):
\begin{eqnarray}
  \mathrm A_{topo} & = & - \frac{1}{2}(1+\cos^2 \frac{\theta}{2}) d\theta - \frac{1}{4\ell} \frac{\mathring L}{R} d\left(u \sin \frac{\theta}{2} \right) + \mathcal O\left(\frac{1}{\ell^2}\right) \\
  & = & - \frac{3}{4} d\theta - \frac{1}{4} d \sin \theta - \frac{1}{4\ell} \frac{\mathring L}{R} d\left(u \sin \frac{\theta}{2} \right)+ \mathcal O\left(\frac{1}{\ell^2}\right)
\end{eqnarray}
For the other quasicoherent state we have $\mathrm A_*=\mathrm A_{geo}+\mathrm A_{def}-\mathrm A_{topo}$.\\

Let $\mathscr C$ be a closed path drawn on the eigensurface $M_\Lambda$, we have
\begin{equation}
  \oint_{\mathscr C} \mathrm A_{topo} = -\frac{3}{2}\pi p - ((-1)^p-1) \frac{1}{4\ell} \frac{\mathring L}{R} u_0 \sin\frac{\theta_0}{2} +  \mathcal O\left(\frac{1}{\ell^2}\right)
\end{equation}
where $p$ is the number of turns of $\mathscr C$ around $0$. $\oint_{\mathscr C} \mathrm A_{topo}$ is well topological because it does not depend of the shape of $\mathscr C$ but only on its number of turns $p$ and on its starting point $(u_0,\theta_0)$. The effect of the ``twisted'' topology of $\mathcal M^2$ is then that the topological phase is non-trivial if and only if the number of turns is odd.

\subsection{Adiabatic transport}
The adiabatic couplings between $|\Lambda\rrangle$ and $|\Lambda_*\rrangle$ are non-zero: $\mathrm C = -\imath \llangle \Lambda_*|d|\Lambda \rrangle = \mathrm C_{geom} + \mathrm C_{topo}$ with (see \ref{C_mobius})
\begin{equation}
  \mathrm C_{geo} = \frac{\imath}{2\ell} \frac{\mathring L}{R} u \cos \theta \cos \frac{\theta}{2} d\theta + \mathcal O\left(\frac{1}{\ell^2}\right)
\end{equation}
and
\begin{eqnarray}
  \mathrm C_{topo} & = & \left(-\frac{e^{-\imath \theta}}{4} \sin \theta+\frac{1}{4} \cos\theta (\imath-\sin \theta) \right)d\theta-\frac{1}{4\ell} \frac{\mathring L}{R} d\left(e^{\imath \theta} u \cos \frac{\theta}{2} \right) + \mathcal O\left(\frac{1}{\ell^2}\right) \\
  & = & \frac{1}{16} d\left(e^{-2\imath \theta}+2\imath \theta-2\sin\theta(-2\imath+\sin\theta)\right) -\frac{1}{4\ell} \frac{\mathring L}{R} d\left(e^{\imath \theta} u \cos \frac{\theta}{2} \right) + \mathcal O\left(\frac{1}{\ell^2}\right)
\end{eqnarray}
The adiabatic transport is then governed by the non-abelian gauge potential:
\begin{equation}
  \pmb {\mathrm A} = \left(\begin{array}{cc} \mathrm A_{geo+}  & - \mathrm C_{geo} \\ \mathrm C_{geo} & \mathrm A_{geo+} \end{array} \right) + \left(\begin{array}{cc} \mathrm A_{topo}  & \overline{\mathrm C_{topo}} \\ \mathrm C_{topo} & -\mathrm A_{topo} \end{array} \right)
\end{equation}
with $\mathrm A_{geo+} = \mathrm A_{geo}+\mathrm A_{def}$, the matrix being written in a basis $(|\circ\rangle,|*\rangle)$ representing the two quasicoherent state channels. Let $\mathscr C$ be a closed path on the eigensurface $M_\Lambda$. The adiabatic transport of $|\Lambda_n(u_0,\theta_0)\rrangle$ along $\mathscr C$ is:
\begin{equation}
  |\Psi(T) \rrangle \simeq \left[\Peg^{-\imath \oint_{\mathscr C} \pmb {\mathrm A}} \right]_{\circ \circ} |\Lambda_{n+p}((-1)^pu_0,\theta_0)\rrangle + \left[\Peg^{-\imath \oint_{\mathscr C} \pmb {\mathrm A}} \right]_{* \circ} |\Lambda_{*,n+p}((-1)^pu_0,\theta_0)\rrangle
\end{equation}
where $\Peg$ is the path-ordered exponential: $\frac{d}{ds} \Peg^{-\imath \int_{x_0}^{x(s)} \pmb{\mathrm A}} = -\imath \pmb{\mathrm A}_i \frac{dx^i}{ds} \Peg^{-\imath \int_{x_0}^{x(s)} \pmb{\mathrm A}}$ (for $s$ a curvilinear coordinates onto $\mathscr C$). $p$ is the number of turns of $\mathscr C$ around $0$.\\

By using the intermediate representation theorem we have
\begin{equation}
  \Peg^{-\imath \oint_{\mathscr C} \pmb {\mathrm A}} = e^{-\imath \oint_{\mathscr C} \mathrm A_{geo+}} e^{-\imath \sigma_3 \oint_{\mathscr C} \mathrm A_{topo}} \Peg^{-\imath \oint_{\mathscr C} \tilde {\pmb {\mathrm C}}}
\end{equation}
with $\tilde{\pmb {\mathrm C}} = e^{\imath \sigma_3 \int_{x_0}^{x(s)} \mathrm A_{topo}} \left({0 \atop \mathrm C} {{\bar {\mathrm C}} \atop 0} \right) e^{-\imath \sigma_3 \int_{x_0}^{x(s)} \mathrm A_{topo}} = \left({0 \atop \tilde {\mathrm C}} {{\bar {\tilde {\mathrm  C}}} \atop 0} \right)$ (with $\tilde {\mathrm C} = e^{- 2\imath\int_{x_0}^{x(s)} \mathrm A_{topo}} \mathrm C$). We have then
\begin{eqnarray}
  |\Psi(T)\rrangle & \simeq & e^{-\imath \oint_{\mathscr C} \mathrm A_{geo+}} e^{-\imath \int_{\mathscr C} \mathrm A_{topo}} \left[\Peg^{-\imath \oint_{\mathscr C} \tilde{\pmb {\mathrm C}}} \right]_{\circ \circ} |\Lambda_{n+p}((-1)^pu_0,\theta_0)\rrangle \nonumber \\
  & & + e^{-\imath \oint_{\mathscr C} \mathrm A_{geo+}} e^{\imath \int_{\mathscr C} \mathrm A_{topo}} \left[\Peg^{-\imath \oint_{\mathscr C} \tilde{\pmb {\mathrm C}}} \right]_{* \circ} |\Lambda_{*,n+p}((-1)^pu_0,\theta_0)\rrangle
\end{eqnarray}

\section{Other examples: noncommutative torus and quantum Klein bottle}
\subsection{The noncommutative torus}
The topological torus is defined by $\mathcal T^2=\mathbb R^2/(2\pi \ell \mathbb Z^2)$. Let $\theta = (\theta^1,\theta^2)$ be coordinates onto $\mathcal T^2$. The analytical functions onto $\mathcal T^2$ are such that $g(\theta) = \sum_{n \in \mathbb Z^2} g_n e^{\imath n \cdot \theta}$ (with $g_n \in \mathbb C$). The standard embedding of the torus into the space $f:\mathcal T^2 \to \mathbb R^3$ is defined by $f(\theta) = (R+r \cos \theta^2)\cos \theta^1, (R+r\cos \theta^2) \sin \theta^1,r\sin \theta^2)$ with $R,r \in \mathbb R^{+*}$ the torus radii.\\
The topological noncommutative torus is then the algebra $\mathfrak X_{\mathcal T^2}$ generated by $\Theta^1 = e^{\imath \frac{a+a^+}{2\ell}}= D(\frac{\imath}{2\ell})$ (quantisation of $e^{\imath \theta^1}$) and $\Theta^2 = e^{\frac{a-a^+}{2\ell}} = D(-\frac{1}{2\ell})$ (quantisation of $e^{\imath \theta^2}$), with $\Theta^1 \Theta^2 = e^{- \frac{\imath}{2\ell^2}} \Theta^2 \Theta^1$.\\
There is an ambiguity of quantisation to defined the embedding of the noncommutative torus, the elements $e^{\imath (\theta^1+\theta^2)}$ being associated to $\Theta^1\Theta^2$ or $\Theta^2\Theta^1$. To solve this ambiguity, we choose to consider the quantisation by the displacement operator:
\begin{equation}
  e^{\imath (\theta^1+\theta^2)} \xrightarrow{quantisation} D\left(\frac{-1+\imath}{2\ell}\right) = e^{\frac{\imath}{4\ell^2}} \Theta^1 \Theta^2 = e^{-\frac{\imath}{4\ell^2}} \Theta^2 \Theta^1
\end{equation}
and then
\begin{eqnarray}
  A & = & RD\left(\frac{\imath}{2\ell}\right) + \frac{r}{2} \left(D\left(\frac{-1+\imath}{2\ell}\right)+D\left(\frac{1+\imath}{2\ell}\right) \right) \\
  X^3 & = & \frac{r}{2\imath} \left(D\left(-\frac{1}{2\ell}\right)-D\left(\frac{1}{2\ell}\right) \right)
\end{eqnarray}
It follows that
\begin{eqnarray}
  \varphi_A(\beta) & = & \left(R_\ell+r_\ell' \cos(\IM(\beta)/\ell)\right)e^{\imath \RE(\beta)/\ell} \\
  \varphi_{X^3}(\beta) & = & r_\ell \sin(\IM(\beta)/\ell)
\end{eqnarray}
with $R_\ell =Re^{\frac{1}{8\ell^2}}$, $r_\ell = re^{\frac{1}{8\ell^2}}$ and $r_\ell' = re^{\frac{1}{4\ell^2}}$. The eigensurface $M_\Lambda$ is then described with
\begin{eqnarray}
  & & \alpha \in \Ran \varphi_A = \{(R_\ell+r'_\ell \cos \theta^2)e^{\imath \theta^1} \}_{\theta \in [0,2\pi[^2} \\
  & \Rightarrow & \alpha_{An} = \ell(\theta^1+2n^1\pi)+\imath \ell(\theta^2+2n^2\pi) \qquad n \in \mathbb Z^2 \\
  & & x^3 = r_\ell \sin \theta^2
\end{eqnarray}    

The density matrix of the quasicoherent state is then
\begin{equation}
  \rho_\Lambda(x) = \frac{1}{2} \left(\begin{array}{cc} 1+\sin^2 \theta^2 & \frac{e^{-\imath \theta^1}}{2} \sin(2\theta^2) \\
    \frac{e^{\imath \theta^1}}{2} \sin(2\theta^2) & \cos^2\theta^2 \end{array} \right)+ \mathcal O\left(\frac{1}{\ell^2}\right)
\end{equation}
The magnetic potential of the noncommutative torus is $\mathrm A=\mathrm A_{geo}+\mathrm A_{def}+\mathrm A_{emb}$ with
\begin{equation}
  \mathrm A_{geo} = \ell(\theta^1+2n^1\pi)d\theta^2-\ell(\theta^2+2n^2\pi)d\theta^1
\end{equation}
\begin{equation}
  \mathrm A_{def} = - \frac{\sin\theta^2}{16r(R+r\cos\theta^2)} d\theta^1 + \mathcal O\left(\frac{1}{\ell}\right)
\end{equation}
\begin{equation}
  \mathrm A_{emb} = \left(-\frac{3}{4}+\frac{1}{4} \cos(2\theta^2)\right) d\theta^1+ \mathcal O\left(\frac{1}{\ell^2}\right)
\end{equation}
and the adiabatic couplings are
\begin{equation}
  \mathrm C = \frac{e^{-\imath \theta^1}}{4} \sin(2\theta^2) d\theta^1 + \frac{1}{2}(e^{-\imath \theta^1} \cos^2\theta^2+e^{\imath \theta^1} \sin^2\theta^2) d\theta^2 + \mathcal O\left(\frac{1}{\ell^2}\right)
\end{equation}

(The calculations are similar to the ones of \ref{calculs_mobius}). The phase issuing from the embedding in $\mathbb R^3$ is here not topological, except for the paths such $\theta^2 = C^{ste}$ (big circles of the torus). The reason is that the curvature of the torus on the surface circumscribing the genus is larger than the one on the opposite surface (the Gaussian curvature of the embedded classical torus being $K(\theta^2) = \frac{\cos\theta^2}{r(R+r\cos\theta^2)}$).

\subsection{The quantum Klein bottle}
The classical Klein bottle can be topologically defined as the torus $\mathcal K^2 = \mathbb R^2/(2\pi\ell \mathbb Z^2)$ but with an algebra of analytical functions satisfying $g(\theta^1+2\pi,\theta^2) = g(\theta^1,-\theta^2)$. It follows that $g(\theta^1,\theta^2) = \sum_{n \in \mathbb Z^2} g_n e^{\imath n_1 \theta^1} (e^{\imath \frac{\theta^1}{2}} e^{\imath \theta^2})^{n_2}$. The ``bagel'' immersion $f:\mathcal K^2\to \mathbb R^3$ of the Klein bottle is defined by:
\begin{equation}
  f(\theta^1,\theta^2) = \left(\begin{array}{c} (R+r \cos \frac{\theta^1}{2} \sin \theta^2 - r \sin \frac{\theta^1}{2} \sin(2\theta^2)) \cos \theta^1 \\
    (R+r \cos \frac{\theta^1}{2} \sin \theta^2 - r \sin \frac{\theta^1}{2} \sin(2\theta^2)) \sin \theta^1 \\
    r \sin \frac{\theta^1}{2} \sin \theta^2 - r \cos \frac{\theta^1}{2} \sin(2\theta^2) \end{array} \right)
\end{equation}
with $R,r \in \mathbb R^{+*}$ the ``bagel'' radii. The topological noncommutative Klein bottle is then the algebra $\mathfrak X_{\mathcal K^2}$ generated by $\Theta^1 = e^{\imath \frac{a+a^+}{2\ell}} = D(\frac{\imath}{2\ell})$ and $\Theta^2 = e^{\imath \frac{a+a^+}{4\ell} + \frac{a-a^+}{2\ell}} = D(-\frac{1}{2\ell}+\frac{\imath}{4\ell})$. 
We can define the operators $A$ and $X^3$ directly from their $|\alpha\rangle$-representations:
\begin{eqnarray}
  \varphi_A(\beta) & = & \left(R+r \cos \frac{\RE\beta}{2\ell} \sin \frac{\IM\beta}{\ell} - r \sin \frac{\RE\beta}{2\ell} \sin \frac{2\IM\beta}{\ell}\right) e^{\imath \frac{\RE\beta}{\ell}} \\
  \varphi_{X^3}(\beta) & = & r \sin \frac{\RE \beta}{2\ell} \sin \frac{\IM\beta}{\ell} - r \cos \frac{\RE\beta}{2\ell} \sin \frac{2\IM\beta}{\ell}
\end{eqnarray}
This example is too complicated to obtain analytical expressions of the quasicoherent states and of the magnetic potentials. The matrix elements of the quasicoherent density matrix are represented fig. \ref{densmat_Klein}.
\begin{figure}
  \includegraphics[width=7cm]{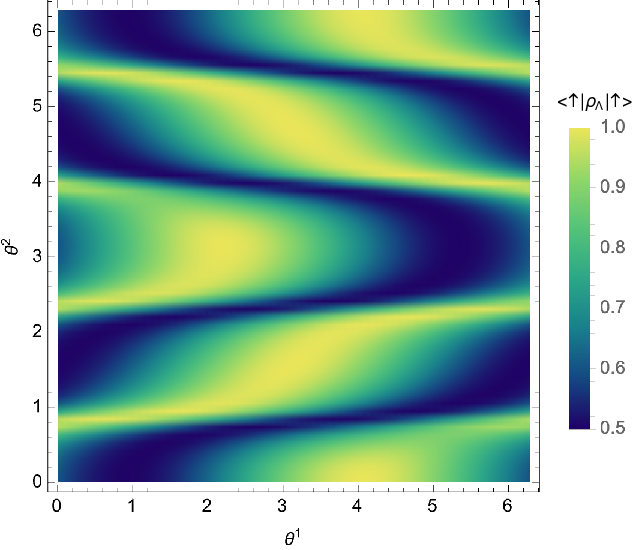}\includegraphics[width=7cm]{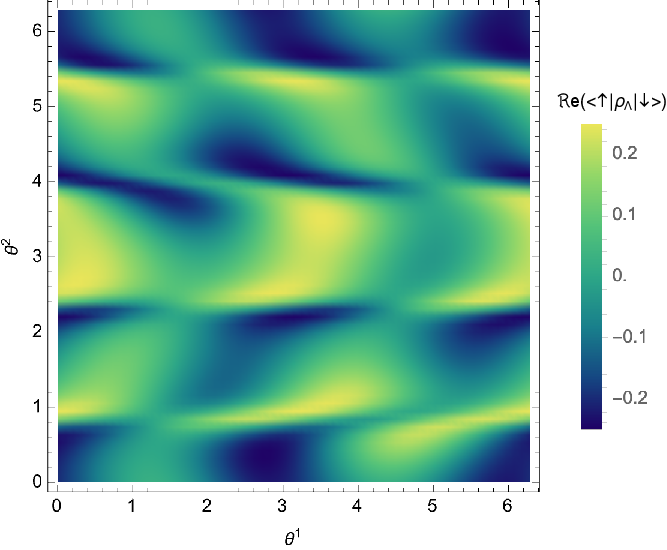} \\
  \caption{\label{densmat_Klein} Matrix elements $\langle \uparrow |\rho_\Lambda|\uparrow \rangle$ (population of the state $\uparrow$) and $\RE(\langle \uparrow |\rho_\Lambda |\downarrow \rangle)$ (coherence) of the quasicoherent density matrix of the quantum Klein bottle with $R=2$, $r=1$ and $\ell=10^{2}$. The integrations have been numerically computed.}
\end{figure}
The magnetic potential is $\mathrm A=\mathrm A_{geo}+\mathrm A_{def}+\mathrm A_{imm}$, with $\mathrm A_{geo} = \ell(\theta^1+2n^1\pi)d\theta^2-\ell(\theta^2+2n^2\pi)d\theta^1$, $\mathrm A_{def}=- \ell \IM(\delta_A) d\theta^1$ ($\IM(\delta_A)$ is represented fig. \ref{deltaA_Klein}); and $\mathrm A_{imm}$ (represented fig. \ref{Aemb_Klein}).
\begin{figure}
  \includegraphics[width=7cm]{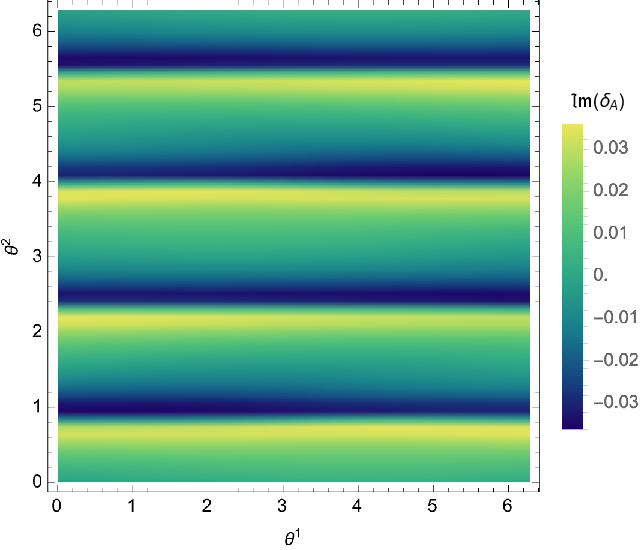}
  \caption{\label{deltaA_Klein} Deformation field $\delta_A = 2\int_{\mathbb C} \frac{\beta|\beta|^2e^{-|\beta|^2}}{|\Delta\varphi_A|^2+|\Delta \varphi_{X^3}|^2} \frac{d^2\beta}{\pi N_x^2}$ of the quantum Klein bottle with $R=2$, $r=1$ and $\ell=10^2$. The integrations have been numerically computed. $\RE(\delta_A)=0$.}
\end{figure}
\begin{figure}
   \includegraphics[width=7cm]{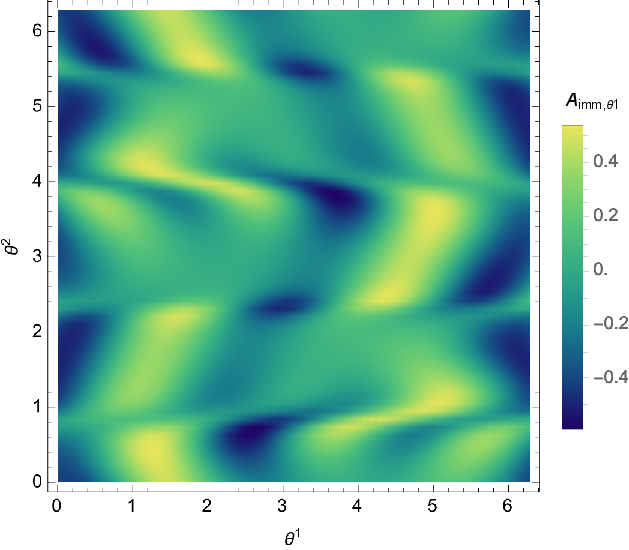}\includegraphics[width=7cm]{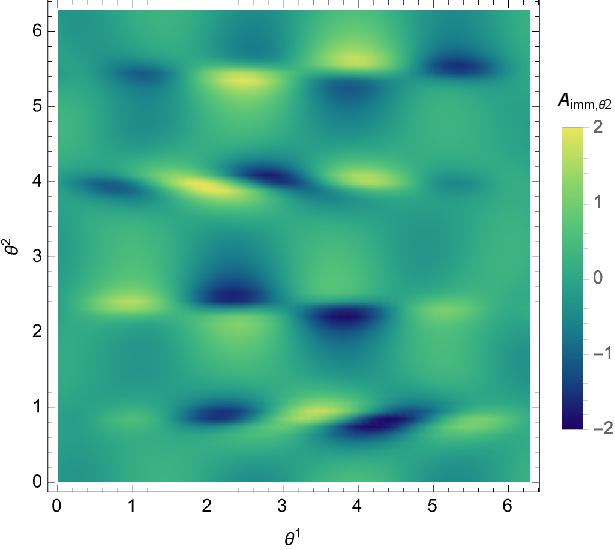} \\
  \caption{\label{Aemb_Klein} The magnetic potential $\mathrm A_{imm} = \mathrm A_{imm,\theta^1}d\theta^1+\mathrm A_{imm,\theta^2}d\theta^2$ issuing from the space immersion of the quantum Klein bottle with $R=2$, $r=1$ and $\ell=10^{2}$. The integrations have been numerically computed.}
\end{figure}

\section{Conclusion}\label{conclusion}
We have find analytical expressions for the quasicoherent states of any CCR D2-brane. But in general these solutions are in an enlarged space, i.e. they are elements of the rigged Fock space $\mathscr F_\infty$ (and not necessarily elements of the Fock space $\mathscr F$). The condition to $|\psi \rangle = \int_{\mathbb C} \psi(\alpha)|\alpha\rangle \frac{d^2\alpha}{\pi} \in \mathscr F_{\infty}$ be in $\mathscr F$, is that $\exists \psi_n \in \mathbb C$ such that $\psi(\alpha) = e^{-|\alpha|^2/2} \sum_{n=0}^{+\infty} \frac{\bar \alpha^n}{\sqrt{n!}} \psi_n$ (in order to $\psi(\alpha)$ be the inner product of a vector of $\mathscr F$ with the coherent state $|\alpha\rangle$). In other words, $|\psi\rangle \in \mathscr F$ if $\alpha \mapsto e^{|\alpha|^2/2} \psi(\alpha)$ is anti-holomorphic at $0$ ($\iff \frac{\partial}{\partial \alpha}(e^{|\alpha|^2/2} \psi(\alpha))=0$). Note that we can have $|\Lambda \rrangle$,$|\Lambda_*\rrangle$ both in $\mathscr F_{\infty} \setminus \mathscr F$, but with some linear combinations of them in $\mathscr F$. We hope that in some cases, their exists at least one vector line in $\Lin(|\Lambda \rrangle$,$|\Lambda_*\rrangle)$ which is in $\mathscr F$, as for the case of the noncommutative plane.\\

We have find the quasicoherent states and the associated magnetic potentials for different CCR D2-branes quantisations of topologically flat surfaces. We have shown that the magnetic potential of the noncommutative cylinder is as if its axis is a line magnetically charged by a density of magnetic monopoles surrounded by a solenoid. The adiabatic transport of a fermionic string along the cylindrical D2-brane exhibits then an Aharonov-Bohm effect associated with this solenoid. Moreover we have shown that the transport of the quasicoherent spin density matrix along a closed path onto a noncommutative M\"obius strip, returns to itself only if the number of turns is even. This is the quantum counterpart of the classical non-orientability of a M\"obius strip (the parallel transport of a normal vector along a closed path onto a M\"obius strip, returns to itself only if the number of turns is even).\\

When the adiabatic couplings are zero $\mathrm C=0$, we have two gauge potentials (two magnetic potentials) $\mathrm A$ and $\mathrm A_*$ corresponding to the two choices of local orientation of the D2-brane (or to the two choices of chirality of the massless fermionic string in the BFSS model). The sources of the associated magnetic fields $\mathrm F = d\mathrm A$ and $\mathrm F_* = d\mathrm A_*$ can be assimilated to magnetic monopoles and electric currents. But in the cases where the adiabic couplings are non-zero, the D2-brane exhibits a non-abelian gauge potential $\pmb{\mathrm A}$ with an associated Yang-Mills field $\pmb{\mathrm F} = d\pmb{\mathrm A}+\pmb{\mathrm A} \wedge \pmb{\mathrm A}$ of gauge symmetry $U(2)$. This is interesting for the interpretation of matrix models in M-theory. Some D2-branes (noncommutative plane, noncommutative cylinder) exhibit only the abelian gauge symmetry group $U(1) \times U(1)$ (by taking into account the orientation reversal) whereas other ones (noncommutative M\"obius strip, noncommutative torus, quantum Klein bottle) exhibit the non-abelian gauge symmetry group $U(2)$. If it the case, we have $|\Psi(T) \rrangle =  \sum_{\bullet \in \{\circ,*\}} [\Peg^{-\imath \oint_{\mathscr C} \pmb{\mathrm A}}]_{\bullet \circ} |\Lambda_\bullet(x(T)) \rrangle$ if $|\Psi(0)\rrangle = |\Lambda_\circ(x(0))\rrangle$, and then $\rho_{\Psi}(T) = \Peg^{-\imath \oint_{\mathscr C} \pmb{\mathrm A}} \rho_\Lambda(x(T)) \left(\Peg^{-\imath \oint_{\mathscr C} \pmb{\mathrm A}}\right)^\dagger$ and not $|\Lambda_\circ(x(T))\rrangle \llangle \Lambda_\circ(x(T))|$. In this meaning we can see the degeneracy of $\ker \slashed D_x$ in $\mathscr F_\infty$ as the manifestation of a decoherence effect in the D2-brane. This effect disappears if $\mathrm C=0$ (in some D2-brane or for some special paths $\mathscr C$).\\

Another important application of this paper in BFSS or BMN matrix models concerns the emergent geometry. Usually to extract the classical geometry underlying the quantized space, it needs to find matrices $X^i_{(N)}$ of size $N \times N$ similar to $X^i$ and to find the semi-classical limit $N \to + \infty$ where $\lim_{N\to +\infty} [X^i_{(N)},X^j_{(N)}]=0$ (or to find a Weyl-Moyal quantisation associated with $\{X^i\}_i$ if the original Hilbert space is still of infinite dimension). But such a limit is not satisfactory from a physical point of view where we want to find the classical geometry closest to the quantum geometry. The quasicoherent states with their eigenmanifolds $M_\Lambda$ are an efficient tool to find the emergent geometry \cite{Schneiderbauer, Steinacker2}. Moreover, the adiabatic limit being similar (in time) to a semi-classical limit, this one provides also important informations (metrics, lorentz connection and emergent space-time fiolation,...) concerning the emergent geometry  \cite{Viennot1, Viennot2}. The study of the emergent geometry requiring the knowledge of the quasicoherent states, analytic formulae for these ones are important results for this task.\\

Throughout the present paper, we have considered that the CCR D2-brane does not evolve (only the probe is moved along the eigensurface), i.e. we have considered that the coordinates observables $X^i$ are time-independent. If we suppose that these ones depend on the time sufficiently slowly to the adiabatic assumption remains valid, we can treat the time dependancy of $X^i$ as a new adiabatic parameter. In other words we can write $\slashed D_{x,t'} = \sigma_i \otimes (X^i(t')-x^i)$ and consider the equation $|\dot \Psi(t) \rrangle = \slashed D_{x(t),t'(t)} |\Psi(t)\rrangle = \sigma_i \otimes (X^i(t)-x^i(t)) |\Psi(t)\rrangle$ with $t'(t)=t$. The quasicoherent states are then dependent on $x$ and $t'$, and the adiabatic potential becomes
\begin{equation}
  \mathrm A(x,t') = \mathrm A_i(x,t') dx^i + \mathrm A_0(x,t') dt'
\end{equation}
with the time-dependent magnetic potential:
\begin{equation}
  \mathrm A_i(x,t') dx^i = -\imath \llangle \Lambda(x,t')|d|\Lambda(x,t')\rrangle
\end{equation}
and the electric potential:
\begin{equation}
  \mathrm A_0(x,t') = -\imath \llangle \Lambda(x,t')|\frac{\partial}{\partial t'} |\Lambda(x,t')\rrangle
\end{equation}
which induces an electric field on the CCR D2-brane $\mathrm E = d \mathrm A_0$ (``$d$'' stands only for the space differential). We can apply this in the case where the time-dependency of $X^i$ follows Heisenberg equations (as for applications in quantum information theory) or follows noncommutative Klein-Gordon equations (as for applications in M-theory).\\

To simplify the analysis, we have focus this paper on CCR D2-branes with a single bosonic mode (associated with the CCR algebra generated by $(a,a^+)$). It is not difficult to apply the same method to CCR D2-branes with multiple bosonic modes (with a CCR algebra generated by $(a_i,a^+_i)_{i=1,...,m}$). For example, consider $\slashed D_x = \sigma_i \otimes (\sum_j X^i_{(j)}-x^i)$ where $X^i_{(j)}$ are operators in $\Env(a_j,a_j^+)$. Let $|\Lambda^{(j)}_\bullet (x)\rrangle$ be the quasicoherent states of the single mode Dirac operator $\sigma_i \otimes (X^i_{(j)}-x^i)$, then we have the following set of quasicoherent states of the D2-brane: $\forall j$, $\forall I$
\begin{equation}\label{QCSmulti}
  |\Lambda^{[j]}_I(x) \rrangle = |\Lambda^{(j)}_{I_j}(x)\rrangle \otimes \bigotimes_{k \not=j} |\Lambda^{(k)}_{I_k}(0) \rrangle
\end{equation}
where $I$ is a set of $m$ symbols $\bullet \in \{\circ,*\}$. For example let a multi-mode noncommutative cylinder be such that $A_{(j)} = R_j e^{\frac{a_j-a_j^\dagger}{2\ell}}$ and $X^3_{(j)} = L_j \frac{a_j+a_j^\dagger}{2}$. If we suppose that $R_j < R_{j+1}$, the eigenmanifold $M_\Lambda$ of the CCR D2-brane is a set of $m$ concentric cylinders of radii $\{R_i\}_{i=1,...,n}$. If we add to $\slashed D_x$ perturbative couplings between the bosonic modes (i.e. for example terms with $\epsilon a_i a_j$, $\epsilon a_i^\dagger a_j$ for $i \not=j$ with $|\epsilon|\ll 1$), we can calculate the quasicoherent states from the expressions eq. (\ref{QCSmulti}) by the quantum perturbation theory. With non-perturbative couplings, it is possible to generalise theorem \ref{Theo} with an integration in $\mathbb C^m$ and Perelomov coherent states $|\alpha_1,...,\alpha_m\rangle = \bigotimes_i|\alpha_i\rangle$. See also \ref{highdim} for a discussion about multi-mode CCR D$p$-branes in higher dimensions.\\

An open question is the possibility to find equivalents of theorem \ref{Theo} for D2-branes being not topologically flat, and so for D2-branes topologically related to the fuzzy sphere or to the fuzzy hyperboloid. Such a result would consist to decompose the quasicoherent states onto respectively the $SU(2)$ and $SU(1,1)$ Perelomov coherent states \cite{Perelomov}. For the fuzzy sphere, this decomposition is known (see for example \cite{Viennot1,Viennot2}). But for a fuzzy ellipsoid for example, it is not clear that we can have a result directly similar to theorem \ref{Theo} with the $SU(2)$ coherent states, because to the knowledge of the author of this paper, an equivalent theorem to the Sudarshan-Mehta theorem \cite{Sudarshan,Mehta} assuring a diagonal representation of any operator, does not exist for the $SU(2)$ coherent states (or for the $SU(1,1)$ coherent states).

\appendix
\section{About the $|\alpha\rangle$-representation}\label{alphaRep}
The diagonal representation $\varphi_X(\alpha)$ of an operator $X$ \cite{Sudarshan, Mehta} can be computed as follows:
\begin{prop}
  The diagonal representations of $X$, $\varphi_X(\alpha)$, and $\langle \alpha|X|\alpha \rangle$ are related by
  \begin{equation}
    \langle \alpha|X|\alpha \rangle = \int_{\mathbb C} \varphi_X(\beta) e^{-|\alpha-\beta|^2} \frac{d^2 \beta}{\pi}
  \end{equation}
  and
  \begin{equation}
    \varphi_X(\alpha) =  e^{- \frac{\partial^2}{\partial \alpha \partial \bar \alpha}} \langle \alpha|X|\alpha \rangle
  \end{equation}
\end{prop}

\begin{proof}
  \begin{eqnarray}
    X = \int_{\mathbb C} \varphi_X(\beta) |\beta \rangle \langle \beta| \frac{d^2\beta}{\pi} \Rightarrow \langle \alpha |X|\alpha \rangle & = & \int_{\mathbb C} \varphi_X(\beta) |\langle \alpha |\beta \rangle|^2 \frac{d^2\beta}{\pi} \\
    & = & \int_{\mathbb C} \varphi_X(\beta) e^{-|\alpha-\beta|^2} \frac{d^2 \beta}{\pi}
  \end{eqnarray}
  Let $G(\alpha,\beta) = e^{-|\alpha-\beta|^2}$ be viewed as a Green function. We search $D_\alpha$ the differential operator such that $D_\alpha G(\alpha,\beta) = \delta^2(\alpha-\beta)$ (where the complex variable Dirac distribution is defined by $\delta^2(z) = \delta(\RE(z))\delta(\IM(z))$. By Fourier transformation of this equation, we have:
  \begin{equation}
    \hat D(k) \frac{e^{-|k|^2/4}}{2} = \frac{1}{2}
  \end{equation}
  where $\hat D$ is the Fourier transformation of $D_\alpha$, $k \in \mathbb C$ is the conjugated variable of $\alpha \in \mathbb C$. We have then
  \begin{equation}
    \hat D(k) = e^{|k|^2/4} \Rightarrow D_\alpha = e^{-\frac{1}{4} \Delta}
  \end{equation}
  where $\Delta = \frac{\partial^2}{\partial \RE(\alpha)^2} + \frac{\partial^2}{\partial \IM(\alpha)^2}$ is the Laplacian of $\mathbb C$. Since $\frac{\partial}{\partial \alpha} = \frac{1}{2}\left(\frac{\partial}{\partial \RE(\alpha)} - \imath \frac{\partial}{\partial \IM(\alpha)} \right)$ we have $\Delta = 4 \frac{\partial^2}{\partial \alpha \partial \bar \alpha}$.
\end{proof}

As shown in \cite{Mehta}, if $X$ is a polynomials of $a$ and $a^+$, we have
\begin{equation}
  \varphi_{\colon X \colon}(\alpha) = {_*}\langle \alpha|;X;|\alpha \rangle_*
\end{equation}
where $\colon X \colon$ is the normal order of $X$ (operators $a^+$ move on the left by the commutation relation $[a,a^+]=1$), $;X;$ is the anti-normal order of $X$ (operators $a^+$ move on the right by the commutation relation), and $|\alpha\rangle_*$ is the singular distribution eigenvector of $a^+$ ($a^+|\alpha\rangle_* = \bar \alpha|\alpha\rangle_*$). The diagonal representations of some operators are shown table \ref{phiX}.

\begin{table}
  \caption{\label{phiX} Diagonal $|\alpha\rangle$-representation of some operators in $\Env(a,a^+)$.}
  \begin{center}
    \begin{tabular}{c|c|c}
      $X$ & $\langle \alpha|\colon X \colon|\alpha\rangle$ & $\varphi_{\colon X \colon}(\alpha) = {_*}\langle \alpha|;X;|\alpha \rangle_* $ \\
      \hline
      $1$ & $1$ & $1$ \\
      $a$ & $\alpha$ & $\alpha$ \\
      $a^+$ & $\bar \alpha$ & $\bar \alpha$ \\
      $a^2$ & $\alpha^2$ & $\alpha^2$ \\
      $a^+a$ & $|\alpha|^2$ & $|\alpha|^2-1$ \\
      $(a^+)^2$ & $\bar \alpha^2$ & $\bar \alpha^2$ \\
      $a^3$ & $\alpha^3$ & $\alpha^3$ \\
      $a^+a^2$ & $|\alpha|^2\alpha$ & $(|\alpha|^2-2)\alpha$ \\
      $(a^+)^2a$ & $|\alpha|^2\bar \alpha$ & $(|\alpha|^2-2)\bar \alpha$ \\
      $(a^+)^3$ & $\bar \alpha^3$ & $\bar \alpha^3$ \\
      $e^{\beta a^+-\bar \beta a}$ & $e^{-|\beta|^2/2}e^{\beta \bar \alpha-\bar \beta \alpha}$ & $e^{+|\beta|^2/2}e^{\beta \bar \alpha-\bar \beta \alpha}$
    \end{tabular}
    \end{center}
\end{table}

\section{Some calculations}
\subsection{Calculation of the magnetic potential of the noncommutative cylinder}\label{CompCylind}
To simply the notation we write:
\begin{equation}
  |\Lambda_n\rrangle = \int_{\mathbb C} \xi(\beta)e^{-\imath u\IM(\beta)}e^{\imath \ell \RE(\beta)(\theta+2n\pi)} \left(\begin{array}{c} \zeta(\beta)e^{-\imath \theta} \\ -L\RE(\beta) \end{array} \right) |\beta+\alpha_A\rangle \frac{d^2\beta}{\pi N}
\end{equation}
We have:
\begin{eqnarray}
  d|\Lambda_n\rrangle & = & -\int_{\mathbb C} \IM(\beta) \xi(\beta)e^{-\imath u\IM(\beta)}e^{\imath \ell \RE(\beta)(\theta+2n\pi)} \left(\begin{array}{c} \zeta(\beta)e^{-\imath \theta} \\ -L\RE(\beta) \end{array} \right) |\beta+\alpha_A\rangle \frac{d^2\beta}{\pi N} du \nonumber \\
    & & +\imath \ell \int_{\mathbb C} \RE(\beta) \xi(\beta)e^{-\imath u\IM(\beta)}e^{\imath \ell \RE(\beta)(\theta+2n\pi)} \left(\begin{array}{c} \zeta(\beta)e^{-\imath \theta} \\ -L\RE(\beta) \end{array} \right) |\beta+\alpha_A\rangle \frac{d^2\beta}{\pi N} d\theta \nonumber \\
    & & -\imath  \int_{\mathbb C} \xi(\beta)e^{-\imath u\IM(\beta)}e^{\imath \ell \RE(\beta)(\theta+2n\pi)} \left(\begin{array}{c} \zeta(\beta)e^{-\imath \theta} \\ 0 \end{array} \right) |\beta+\alpha_A\rangle \frac{d^2\beta}{\pi N} d\theta \nonumber \\
    & & + \int_{\mathbb C} \xi(\beta)e^{-\imath u\IM(\beta)}e^{\imath \ell \RE(\beta)(\theta+2n\pi)} \left(\begin{array}{c} \zeta(\beta)e^{-\imath \theta} \\ -L\RE(\beta) \end{array} \right) d|\beta+\alpha_A\rangle \frac{d^2\beta}{\pi N}
\end{eqnarray}
It follows that
\begin{eqnarray}
  \llangle \Lambda_n|d|\Lambda_n\rrangle & = & -\int_{\mathbb C} \IM(\beta)|\xi(\beta)|^2(|\zeta(\beta)|^2+|L\RE(\beta)|^2) \frac{d^2\beta}{\pi N^2} du \nonumber \\
  & & +\imath \ell \int_{\mathbb C} \RE(\beta)|\xi(\beta)|^2(|\zeta(\beta)|^2+|L\RE(\beta)|^2) \frac{d^2\beta}{\pi N^2} d\theta \nonumber \\
  & & -\imath \int_{\mathbb C} |\xi(\beta) \zeta(\beta)|^2 \frac{d^2\beta}{\pi N^2} d\theta \nonumber \\
  & & + Y
\end{eqnarray}
The two first integrals are zero (integration of odd functions). For the computation of last term $Y$ we note that
\begin{eqnarray}
  d|\beta+\alpha_A\rangle & = & \left.\frac{\partial|\beta\rangle}{\partial \beta}\right|_{\beta+\alpha_A} \left(\frac{\partial \alpha_A}{\partial \theta} d\theta+\frac{\partial \alpha_A}{\partial u} du \right) + \left.\frac{\partial|\beta\rangle}{\partial \bar \beta}\right|_{\beta+\alpha_A} \left(\frac{\partial \bar \alpha_A}{\partial \theta} d\theta+\frac{\partial \bar \alpha_A}{\partial u} du \right) \\
  & = & -(\IM(\beta) \ell d\theta+\RE(\beta) du)|\beta+\alpha_A\rangle \nonumber \\
  & & -(\ell^2(\theta+2n\pi) d\theta + udu) |\beta+\alpha_A \rangle \nonumber \\
  & & +(\imath \ell d\theta+du)a^+|\beta+\alpha_A\rangle
\end{eqnarray}
It follows
\begin{eqnarray}
  Y & = & - \int_{\mathbb C} (\IM(\beta) \ell d\theta+\RE(\beta) du) |\xi(\beta)|^2(|\zeta(\beta)|^2+|L\RE(\beta)|^2) \frac{d^2\beta}{\pi N^2} \nonumber \\
  & & - (\ell^2(\theta+2n\pi) d\theta + udu) \int_{\mathbb C} |\xi(\beta)|^2(|\zeta(\beta)|^2+|L\RE(\beta)|^2) \frac{d^2\beta}{\pi N^2} \nonumber \\
  & & + (\imath \ell d\theta+du) \int_{\mathbb C} (\bar \beta + \bar \alpha_A) |\xi(\beta)|^2(|\zeta(\beta)|^2+|L\RE(\beta)|^2) \frac{d^2\beta}{\pi N^2}
\end{eqnarray}
The first integral and the term with $\bar \beta$ in the last integral are zero (integration of odd function). The second integral is one ($\llangle \Lambda|\Lambda\rrangle=1$), it follows that
\begin{eqnarray}
  Y & = & - (\ell^2(\theta+2n\pi) d\theta + udu) + (\imath \ell d\theta+du)(u-\imath \ell(\theta+2n\pi)) \\
  & = & \imath \ell (u d\theta - (\theta+2n\pi)du)
\end{eqnarray}
Finally
\begin{equation}
  \llangle \Lambda|d|\Lambda \rrangle = -\imath \int_{\mathbb C} |\xi(\beta) \zeta(\beta)|^2 \frac{d^2\beta}{\pi N^2} d\theta  + \imath \ell (u d\theta - (\theta+2n\pi)du)
\end{equation}

\subsection{Calculations for the noncommutative M\"obius strip}\label{calculs_mobius}
For the noncommutative M\"obius strip we have
\begin{eqnarray}
  \Delta \varphi_A & = & \left(R+\frac{\mathring L}{\ell} (\RE \beta +u) \cos\left(\frac{\theta}{2}+\frac{\IM \beta}{2\ell}\right)\right) e^{\imath(\theta+\IM \beta/\ell)} \nonumber \\
  & & \qquad - \left(R+\frac{\mathring L}{\ell} u \cos \frac{\theta}{2} \right) e^{\imath \theta} \\
  \Delta \varphi_{X^3} & = & \frac{\mathring L}{\ell} (\RE \beta+u) \sin\left(\frac{\theta}{2}+\frac{\IM \beta}{2\ell}\right) - \frac{\mathring L}{\ell} u \sin \frac{\theta}{2}
\end{eqnarray}
By Taylor expansions we have
\begin{eqnarray}
  \Delta \varphi_A & = & e^{\imath \theta} \left(\mathring L \cos \frac{\theta}{2} \frac{\RE \beta}{\ell} + \imath R \frac{\IM \beta}{\ell} \right. \nonumber \\
  & & \left. - R \frac{\IM^2 \beta}{2\ell^2} - \mathring L (\RE\beta+u) \sin\frac{\theta}{2} \frac{\IM \beta}{2\ell^2} + \imath \mathring L (\RE \beta+u) \cos \frac{\theta}{2} \frac{\IM\beta}{\ell^2} \right) \nonumber \\
  & & + \mathcal O\left(\frac{1}{\ell^3}\right) 
\end{eqnarray}
and
\begin{equation}
  \Delta \varphi_{X^3} = \mathring L \sin \frac{\theta}{2} \frac{\RE \beta}{\ell} + \mathring L (\RE \beta+u) \cos\frac{\theta}{2} \frac{\IM\beta}{2\ell^2} + \mathcal O\left(\frac{1}{\ell^3}\right)
\end{equation}
  
\subsubsection{Normalisation factor}
From the Taylor expansions we have
\begin{eqnarray}
  |\Delta \varphi_A|^2 & = & \mathring L^2 \cos^2 \frac{\theta}{2} \frac{\RE^2\beta}{\ell^2} + R^2 \frac{\IM^2\beta}{\ell^2} \nonumber \\
  & & +\mathring LR (\RE\beta+2u) \cos \frac{\theta}{2} \frac{\IM^2\beta}{\ell^3} - \mathring L^2(\RE\beta+u)\sin \theta \frac{\RE\beta\IM\beta}{2\ell^3} \nonumber \\
  & & + \mathcal O\left(\frac{1}{\ell^4}\right)
\end{eqnarray}
\begin{equation}
  |\Delta \varphi_{X^3}|^2 = \mathring L^2 \sin^2 \frac{\theta}{2} \frac{\RE^2\beta}{\ell^2} + \mathring L^2(\RE\beta+u)\sin \theta \frac{\RE\beta\IM\beta}{2\ell^3} + \mathcal O\left(\frac{1}{\ell^4}\right)
\end{equation}
It follows that
\begin{equation}
  \frac{1}{|\Delta \varphi_A|^2+|\Delta \varphi_{X^3}|^2} = \frac{\ell^2}{\mathring L^2 \RE^2\beta+R^2\IM^2\beta} - \frac{\mathring LR (\RE\beta+2u) \cos \frac{\theta}{2} \IM^2\beta}{(\mathring L^2 \RE^2\beta+R^2\IM^2\beta)^2} \ell + \mathcal O(1)
\end{equation}
Let $d\mu(\beta) = |\beta|^2 e^{-|\beta|^2} \frac{d^2\beta}{\pi}$ be the measure onto $\mathbb C$ associated with the $\beta$-representation. The normalisation factor is (see table \ref{gaussian})
\begin{eqnarray}
  N_x^2 & = & \int_{\mathbb C} \frac{d\mu(\beta)}{|\Delta \varphi_A|^2+|\Delta \varphi_{X^3}|^2} \\
  & = & \frac{\ell^2}{\mathring LR} - \frac{1}{R^2} u \cos \frac{\theta}{2} \ell + \mathcal O(1)
\end{eqnarray}
and then
\begin{equation}
  \frac{1}{N_x^2} = \frac{\mathring L R}{\ell^2} + \mathring L^2 u \cos \frac{\theta}{2} \frac{1}{\ell^3} + \mathcal O \left(\frac{1}{\ell^4}\right)
\end{equation}

\subsubsection{Quasicoherent density matrix} \label{rho_mobius}
We denote by $\odd$ any function of $\beta$ odd with respect to $\RE\beta$ or to $\IM\beta$.\\

\begin{equation}
  \frac{1}{(|\Delta \varphi_A|^2+|\Delta \varphi_{X^3}|^2)^2} = \frac{\ell^4}{\mathring L^2 \RE^2\beta+R^2\IM^2\beta} - \frac{2\mathring LR (\RE\beta+2u) \cos \frac{\theta}{2} \IM^2\beta}{(\mathring L^2 \RE^2\beta+R^2\IM^2\beta)^3} \ell^3 + \mathcal O(\ell^2)
\end{equation}
We have then
\begin{eqnarray}
  \frac{|\Delta \varphi_A|^2}{(|\Delta \varphi_A|^2+|\Delta \varphi_{X^3}|^2)^2} \frac{1}{N_x^2} & = & \mathring LR \frac{\mathring L^2 \cos^2\frac{\theta}{2} \RE^2\beta+R^2\IM^2\beta}{(\mathring L^2 \RE^2\beta+R^2\IM^2\beta)^2} \nonumber \\
  & & + \mathring L^2 u \cos \frac{\theta}{2} \frac{\mathring L^2 \cos^2\frac{\theta}{2} \RE^2\beta+R^2\IM^2\beta}{(\mathring L^2 \RE^2\beta+R^2\IM^2\beta)^2} \frac{1}{\ell} \nonumber \\
  & & -4 \mathring L^2R^2 u \cos \frac{\theta}{2} \IM^2\beta \frac{\mathring L^2 \cos^2\frac{\theta}{2} \RE^2\beta+R^2\IM^2\beta}{(\mathring L^2 \RE^2\beta+R^2\IM^2\beta)^3} \frac{1}{\ell} \nonumber \\
  & & + \frac{2\mathring L^2R^2 u \cos \frac{\theta}{2} \IM^2\beta}{(\mathring L^2 \RE^2\beta+R^2\IM^2\beta)^2} \frac{1}{\ell} \nonumber \\
  & & + \odd + \mathcal O\left(\frac{1}{\ell^2} \right)
\end{eqnarray}
and then (see table \ref{gaussian})
\begin{eqnarray}
  \rho_{\Lambda \uparrow \uparrow} & = & \int_{\mathbb C}  \frac{|\Delta \varphi_A|^2}{(|\Delta \varphi_A|^2+|\Delta \varphi_{X^3}|^2)^2} \frac{d\mu(\beta)}{N_x^2} \\
  & = & \frac{1}{2} (1+\cos^2 \frac{\theta}{2}) + \mathcal O\left(\frac{1}{\ell^2} \right)
\end{eqnarray}

Moreover we have $\rho_{\Lambda \downarrow \downarrow} = 1-\rho_{\Lambda \uparrow \uparrow} = \frac{1}{2} \sin^2 \frac{\theta}{2} + \mathcal O\left(\frac{1}{\ell^2} \right)$.\\

For the off-diagonal part we have
\begin{eqnarray}
  e^{-\imath \theta} \Delta \varphi_A \Delta \varphi_{X^3} & = & \mathring L^2 \sin \theta \frac{\RE^2\beta}{2\ell^2} + \imath \mathring LR \sin \frac{\theta}{2} \frac{\RE\beta \IM\beta}{\ell^2} \nonumber \\
  & & + \mathring L^2 \cos \theta(\RE\beta+u) \frac{\RE\beta\IM\beta}{2\ell^3} - \mathring LR \sin \frac{\theta}{2} \frac{\IM^2\beta \RE\beta}{2\ell^3} \nonumber \\
  & & +\imath \mathring LR \cos\frac{\theta}{2}(\RE\beta+u) \frac{\IM^2\beta}{2\ell^3} + \imath \mathring L^2 \sin \theta(\RE\beta+u) \frac{\IM\beta\RE\beta}{2\ell^3} \nonumber \\
  & & + \mathcal O\left(\frac{1}{\ell^4} \right)
\end{eqnarray}
It follows that
\begin{eqnarray}
  \frac{e^{-\imath \theta} \Delta \varphi_A \Delta \varphi_{X^3}}{(|\Delta \varphi_A|^2+|\Delta \varphi_{X^3}|^2)^2} \frac{1}{N_x^2} & = & \frac{\mathring L^3 R}{2} \frac{\sin \theta \RE^2\beta}{(\mathring L^2 \RE^2\beta+R^2\IM^2\beta)^2} \nonumber \\
  & & + \frac{\mathring L^4}{2} \frac{u \cos \frac{\theta}{2} \sin \theta \RE^2\beta}{(\mathring L^2 \RE^2\beta+R^2\IM^2\beta)^2} \frac{1}{\ell} \nonumber \\
  & & -\frac{2\mathring L^4R^2 u \cos \frac{\theta}{2}\sin \theta \IM^2\beta \RE^2\beta}{(\mathring L^2 \RE^2\beta+R^2\IM^2\beta)^3} \frac{1}{\ell} \nonumber \\
  & & +\imath \frac{\mathring L^2R^2}{2} \frac{u \cos \frac{\theta}{2} \IM^2\beta}{(\mathring L^2 \RE^2\beta+R^2\IM^2\beta)^2} \frac{1}{\ell} \nonumber \\
  & & + \odd + \mathcal O\left(\frac{1}{\ell^2} \right)
\end{eqnarray}
and then (see table \ref{gaussian})
\begin{eqnarray}
  \rho_{\Lambda \downarrow \uparrow} & = & -\int_{\mathbb C}  \frac{\Delta \varphi_A \Delta \varphi_{X^3}}{(|\Delta \varphi_A|^2+|\Delta \varphi_{X^3}|^2)^2} \frac{d\mu(\beta)}{N_x^2} \\
  & = & -\frac{1}{4} e^{\imath \theta} \sin \theta - \frac{\imath}{4\ell} \frac{\mathring L}{R} e^{\imath \theta} u \cos \frac{\theta}{2} + \mathcal O\left(\frac{1}{\ell^2} \right)
\end{eqnarray}
Moreover $\rho_{\Lambda \uparrow \downarrow}  = \overline{\rho_{\Lambda \downarrow \uparrow}}$.

\subsubsection{Deformation magnetic potential} \label{Adef_mobius}
We have
\begin{equation}
  \frac{\beta}{|\Delta \varphi_A|^2+|\Delta \varphi_{X^3}|^2} \frac{1}{N_x^2} = - \frac{\mathring L^2R^2\cos \frac{\theta}{2} \RE^2\beta \IM^2\beta}{(\mathring L^2 \RE^2\beta+R^2\IM^2\beta)^2} \frac{1}{\ell} + \odd + \mathcal O\left(\frac{1}{\ell^2} \right)
\end{equation}
It follows that (see table \ref{gaussian})
\begin{eqnarray}
  \delta_A & = & 2 \int_{\mathbb C} \frac{\beta}{|\Delta \varphi_A|^2+|\Delta \varphi_{X^3}|^2} \frac{d\mu(\beta)}{N_x^2} \\
  & = & - \frac{1}{\ell} \frac{\mathring LR}{(\mathring L+R)^2} \cos \frac{\theta}{2} + \mathcal O\left(\frac{1}{\ell^2} \right)
\end{eqnarray}
and
\begin{eqnarray}
  \mathrm A_{def} & = & -\frac{\imath}{2}(\bar \delta_A d\alpha_A - \delta_A d\bar \alpha_A) \\
  & = & \delta_A d\IM \alpha_A \\
  & = & - \frac{\mathring LR}{(\mathring L+R)^2} \cos \frac{\theta}{2} d\theta + \mathcal O\left(\frac{1}{\ell} \right)
\end{eqnarray}

\subsubsection{Embedding magnetic potential} \label{Atopo_mobius}
Firstly we have
\begin{eqnarray}
  e^{-\imath \theta} d\Delta \varphi_A & = & \imath e^{-\imath \theta} \Delta \varphi_A d\theta \nonumber \\
  & & -\mathring L \sin \frac{\theta}{2} \frac{\RE\beta}{2\ell} d\theta - \mathring L(\RE\beta+u)\cos \frac{\theta}{2} \frac{\IM\beta}{4\ell^2} d\theta \nonumber \\
  & & - \imath \mathring L(\RE\beta+u) \sin\frac{\theta}{2} \frac{\IM\beta}{2\ell^2} d\theta - \mathring L \sin \frac{\theta}{2} \frac{\IM\beta}{2\ell^2} du \nonumber \\
  & & +\imath \mathring L \cos \frac{\theta}{2} \frac{\IM\beta}{\ell^2} du + \mathcal O\left(\frac{1}{\ell^3}\right)
\end{eqnarray}
and
\begin{eqnarray}
  d\Delta \varphi_{X^3} & = & \mathring L \cos \frac{\theta}{2} \frac{\RE\beta}{2\ell} d\theta - \mathring L (\RE\beta+u) \sin \frac{\theta}{2} \frac{\IM\beta}{4\ell^2} \nonumber \\
  & & +\mathring L \cos \frac{\theta}{2} \frac{\IM\beta}{2\ell^2} du + \mathcal O\left(\frac{1}{\ell^3}\right)
\end{eqnarray}
We have then
\begin{eqnarray}
  \overline{\Delta \varphi_A} d\Delta \varphi_A & = & \imath \mathring L^2 \cos^2\frac{\theta}{2} \frac{\RE^2\beta}{\ell^2}d\theta + \imath R^2\frac{\IM^2\beta}{\ell^2}d\theta \nonumber \\
  & & +9\imath \mathring LR u \cos \frac{\theta}{2} \frac{\IM^2\beta}{4\ell^3} d\theta - \mathring L^2 \sin \theta \frac{\RE^2\beta}{4\ell^2} d\theta\nonumber \\
  & &  - \mathring LR u \sin \frac{\theta}{2} \frac{\IM^2\beta}{2\ell^3} d\theta +\imath \mathring LR\sin \frac{\theta}{2} \frac{\IM^2\beta}{2\ell^3} du \nonumber \\
  & & + \mathring LR \cos \frac{\theta}{2} \frac{\IM^2\beta}{\ell^3} du + \odd + \mathcal O\left(\frac{1}{\ell^4}\right)
\end{eqnarray}
and
\begin{equation}
  \Delta \varphi_{X^3} d\Delta \varphi_{X^3} = \mathring L^2 \sin \theta \frac{\RE^2\beta}{4\ell^2} d\theta + \odd + \mathcal O\left(\frac{1}{\ell^4}\right)
\end{equation}
Let $\mathscr A = \imath \frac{\overline{\Delta \varphi_A} d\Delta \varphi_A+\Delta \varphi_{X^3} d\Delta \varphi_{X^3}}{(|\Delta \varphi_A|^2+|\Delta \varphi_{X^3}|^2)^2 N_x^2}= \mathscr A_\theta d\theta + \mathscr A_u du$. We have
\begin{eqnarray}
  \mathscr A_\theta & = & - \frac{\mathring L^3R \cos^2\frac{\theta}{2} \RE^2\beta}{(\mathring L^2 \RE^2\beta+R^2\IM^2\beta)^2} - \frac{\mathring LR^3 \IM^2\beta}{(\mathring L^2 \RE^2\beta+R^2\IM^2\beta)^2} \nonumber \\
  & & -9 \frac{\mathring L^2R^2 u \cos \frac{\theta}{2} \IM^2\beta}{(\mathring L^2 \RE^2\beta+R^2\IM^2\beta)^2} \frac{1}{4\ell} \nonumber \\
  & & -\imath \frac{\mathring L^2R^2u\sin \frac{\theta}{2} \IM^2\beta}{(\mathring L^2 \RE^2\beta+R^2\IM^2\beta)^2} \frac{1}{2\ell} + 4 \frac{\mathring L^4 R^2 u \cos^3 \frac{\theta}{2} \IM^2\beta \RE^2\beta}{(\mathring L^2 \RE^2\beta+R^2\IM^2\beta)^3} \frac{1}{\ell} \nonumber \\
  & & +4 \frac{\mathring L^2 R^4 u \cos \frac{\theta}{2} \IM^4\beta}{(\mathring L^2 \RE^2\beta+R^2\IM^2\beta)^3} \frac{1}{\ell} - \frac{\mathring L^4 u \cos^3 \frac{\theta}{2} \RE^2\beta}{(\mathring L^2 \RE^2\beta+R^2\IM^2\beta)^2} \frac{1}{\ell} \nonumber \\
  & & - \frac{\mathring L^2R^2 u \cos \frac{\theta}{2} \IM^2\beta}{(\mathring L^2 \RE^2\beta+R^2\IM^2\beta)^2} \frac{1}{\ell} + \odd + \mathcal O\left(\frac{1}{\ell^2}\right)
\end{eqnarray}
It follows that (see table \ref{gaussian})
\begin{eqnarray}
  \tilde {\mathrm A}_{topo,\theta} & = & \int_{\mathbb C} \mathscr A_\theta d\mu(\beta) \\
  & = & - \frac{1}{2}(1+\cos^2 \frac{\theta}{2}) - \frac{1}{4\ell} \frac{\mathring L}{R} (\frac{1}{2} u\cos \frac{\theta}{2} + \imath u \sin \frac{\theta}{2}) + \mathcal O\left(\frac{1}{\ell^2}\right)
\end{eqnarray}
In a same way
\begin{eqnarray}
  \mathscr {\mathrm A}_u & = & -\frac{\mathring L^2R^2 \sin \frac{\theta}{2} \IM^2\beta}{(\mathring L^2 \RE^2\beta+R^2\IM^2\beta)^2} \frac{1}{2\ell} + \imath \frac{\mathring L^2R^2 \cos \frac{\theta}{2} \IM^2\beta}{(\mathring L^2 \RE^2\beta+R^2\IM^2\beta)^2} \frac{1}{\ell} \nonumber \\
  & & + \odd + \mathcal O\left(\frac{1}{\ell^2}\right)
\end{eqnarray}
and (see table \ref{gaussian})
\begin{eqnarray}
  \tilde {\mathrm A}_{topo,u} & = & \int_{\mathbb C} \mathscr {\mathrm A}_u d\mu(\beta) \\
  & = & -\frac{1}{4\ell} \frac{\mathring L}{R} \sin \frac{\theta}{2} + \imath \frac{1}{2\ell} \frac{\mathring L}{R} \cos \frac{\theta}{2} + \mathcal O\left(\frac{1}{\ell^2}\right)
\end{eqnarray}
Moreover we have
\begin{equation}
  \imath d \ln N_x = -\imath\frac{1}{2\ell} \frac{\mathring L}{R} \cos \frac{\theta}{2}du + \imath \frac{1}{4\ell} \frac{\mathring L}{R} \sin \frac{\theta}{2} d\theta+ \mathcal O\left(\frac{1}{\ell^2}\right)
\end{equation}
To conclude
\begin{eqnarray}
  \mathrm A_{topo} & = & \tilde {\mathrm A}_{topo,\theta}d\theta+\tilde {\mathrm A}_{topo,u} du + \imath d\ln N_x \\
  & = & -\frac{1}{2}(1+\cos^2\frac{\theta}{2}) d\theta - \frac{1}{4\ell} \frac{\mathring L}{R} \underbrace{(\frac{1}{2} u\cos \frac{\theta}{2} d\theta+ \sin \frac{\theta}{2}du)}_{d(u \sin \frac{\theta}{2})}+ \mathcal O\left(\frac{1}{\ell^2}\right)
\end{eqnarray}

\subsubsection{Adiabatic couplings} \label{C_mobius}
\begin{eqnarray}
  e^{\imath \theta} \Delta \varphi_{X^3} d\overline{\Delta \varphi_A} & = & -\imath \mathring L^2 \sin \theta \frac{\RE^2\beta}{2\ell^2} d\theta - \mathring LR u \cos \frac{\theta}{2} \frac{\IM^2\beta}{2\ell^3} d\theta \nonumber \\
  & & - \mathring L^2 \sin^2\frac{\theta}{2} \frac{\RE^2\beta}{2\ell^2}d\theta + \odd + \mathcal O\left(\frac{1}{\ell^4}\right)
\end{eqnarray}
then
\begin{eqnarray}
  {\scriptstyle  \frac{e^{\imath \theta} \Delta \varphi_{X^3} d\overline{\Delta \varphi_A}}{(|\Delta \varphi_A|^2+|\Delta \varphi_{X^3}|^2)^2} \frac{1}{N_x^2}} & = & -\imath \frac{\mathring L^3R \sin \theta \RE^\beta d\theta}{2(\mathring L^2 \RE^2\beta+R^2\IM^2\beta)^2} - \frac{\mathring L^3R \sin^2\frac{\theta}{2} \RE^2\beta}{2(\mathring L^2 \RE^2\beta+R^2\IM^2\beta)^2 d\theta} \nonumber \\
  & & - \frac{\mathring L^2R^2 u \cos \frac{\theta}{2} \IM^2\beta d\theta}{(\mathring L^2 \RE^2\beta+R^2\IM^2\beta)^2} \frac{1}{2\ell} + \imath \frac{2\mathring L^4R^2 u \cos \frac{\theta}{2} \sin \theta \IM^2\beta \RE^2\beta d\theta}{(\mathring L^2 \RE^2\beta+R^2\IM^2\beta)^3} \frac{1}{\ell} \nonumber \\
  & & - \imath \frac{\mathring L^4u \cos \frac{\theta}{2} \sin \theta \RE^2\beta d\theta}{(\mathring L^2 \RE^2\beta+R^2\IM^2\beta)^2} \frac{1}{2\ell} + \odd + \mathcal O\left(\frac{1}{\ell^2}\right)
\end{eqnarray}
and then (see table \ref{gaussian})
\begin{eqnarray}
  \mathrm C_1 & = & -\imath \int_{\mathbb C} \frac{\Delta \varphi_{X^3} d\overline{\Delta \varphi_A}}{(|\Delta \varphi_A|^2+|\Delta \varphi_{X^3}|^2)^2} \frac{d\mu(\beta)}{N_x^2} \\
  & = & \left(-\frac{1}{4} \sin \theta + \frac{\imath}{4} \sin^2 \frac{\theta}{2} + \frac{\imath}{4\ell} \frac{\mathring L}{R} u \cos \frac{\theta}{2}\right) e^{-\imath \theta} d\theta + \mathcal O\left(\frac{1}{\ell^2}\right)
\end{eqnarray}
In a same way
\begin{eqnarray}
  e^{-\imath \theta} \Delta \varphi_A d\Delta \varphi_{X^3} & = & \mathring L^2 \cos^2\frac{\theta}{2} \frac{\RE^2\beta}{2\ell^2} d\theta - \imath \mathring LR u \sin \frac{\theta}{2} \frac{\IM^2\beta}{4\ell^3} d\theta \nonumber \\
  & & + \imath \mathring LR \cos \frac{\theta}{2} \frac{\IM^2\beta}{2\ell^3} du + \odd + \mathcal O\left(\frac{1}{\ell^4}\right)
\end{eqnarray}
then
\begin{eqnarray}
 {\scriptstyle \frac{e^{-\imath \theta} \Delta \varphi_A d\Delta \varphi_{X^3}}{(|\Delta \varphi_A|^2+|\Delta \varphi_{X^3}|^2)^2} \frac{1}{N_x^2}} & = & \frac{\mathring L^3R \cos^2\frac{\theta}{2} \RE^2\beta}{2(\mathring L^2 \RE^2\beta+R^2\IM^2\beta)^2} d\theta - \imath \frac{\mathring L^2R^2 u \sin \frac{\theta}{2} \IM^2\beta}{(\mathring L^2 \RE^2\beta+R^2\IM^2\beta)^2} \frac{d\theta}{4\ell} \nonumber \\
  & & + \imath \frac{\mathring L^2R^2 \cos \frac{\theta}{2} \IM^2\beta}{(\mathring L^2 \RE^2\beta+R^2\IM^2\beta)^2} \frac{du}{2\ell} - \frac{2\mathring L^4R^2 u \cos^3 \frac{\theta}{2} \IM^2\beta \RE^2\beta}{(\mathring L^2 \RE^2\beta+R^2\IM^2\beta)^3} \frac{d\theta}{\ell} \nonumber \\
  & & + \frac{\mathring L^4u \cos^3\frac{\theta}{2} \RE^2\beta}{(\mathring L^2 \RE^2\beta+R^2\IM^2\beta)^2} \frac{d\theta}{2\ell} + \odd + \mathcal O\left(\frac{1}{\ell^2}\right)
\end{eqnarray}
and then (see table \ref{gaussian})
\begin{eqnarray}
   \mathrm C_2 & = & \imath \int_{\mathbb C} \frac{\Delta \varphi_A d\Delta \varphi_{X^3}}{(|\Delta \varphi_A|^2+|\Delta \varphi_{X^3}|^2)^2} \frac{d\mu(\beta)}{N_x^2} \\
   & = & \left(\frac{\imath}{4} \cos^2\frac{\theta}{2} + \frac{1}{8\ell}\frac{\mathring L}{R} u \sin \frac{\theta}{2} \right) e^{\imath \theta}d\theta \nonumber \\
   & & - \frac{1}{4\ell} \frac{\mathring L}{R} \cos \frac{\theta}{2} e^{\imath \theta} du + \mathcal O\left(\frac{1}{\ell^2}\right)
\end{eqnarray}
Finally we have
\begin{eqnarray}
  \mathrm C & = & \mathrm C_1+\mathrm C_2 \\
  & = & \left(-\frac{e^{-\imath \theta}}{4} \sin \theta+\frac{\imath}{4}(e^{-\imath \theta}\sin^2\frac{\theta}{2}+e^{\imath \theta}\cos^2\frac{\theta}{2}) \right) d\theta \nonumber \\
  & & -\frac{1}{4\ell} \frac{\mathring L}{R} d\left(e^{\imath \theta}u\cos\frac{\theta}{2}\right) + \frac{\imath}{4\ell} \frac{\mathring L}{R} u \underbrace{(e^{-\imath \theta}+e^{\imath \theta})}_{2 \cos\theta} \cos\frac{\theta}{2} d\theta \nonumber \\
  & & +  \mathcal O\left(\frac{1}{\ell^2}\right)
\end{eqnarray}

\section{Complex gaussian integrals}
\subsection{Convergence of complex gaussian integrals}\label{convergence}
The $|\alpha\rangle$-representation of the quasicoherent states induces some formulae implying complex gaussian integrals as:
\begin{equation}
  \int_{\mathbb C} f(\beta) |\beta|^2 e^{-|\beta|^2} \frac{d^2\beta}{\pi}
\end{equation}
with $f(\beta) = \frac{1}{|\Delta \varphi_A(\beta)|^2+|\Delta \varphi_{X^2}(\beta)|^2}$ (for the normalisation of the quasicoherent state), $f(\beta) = \frac{\Delta \varphi_X(\beta)\overline{\Delta \varphi_Y(\beta)}}{(|\Delta \varphi_A(\beta)|^2+|\Delta \varphi_{X^2}(\beta)|^2)^2}$ with $X,Y\in\{A,X^3\}$ (for the density matrix or the magnetic potential), or $f(\beta) = \frac{\beta}{|\Delta \varphi_A(\beta)|^2+|\Delta \varphi_{X^2}(\beta)|^2}$ (for the deformation parameter $\delta_A$). Let $\beta = re^{\imath \theta}$ we have
\begin{equation}
  \int_{\mathbb C} f(\beta) |\beta|^2 e^{-|\beta|^2} \frac{d^2\beta}{\pi} = \frac{1}{\pi} \int_0^{2\pi} \int_0^{+\infty} f(re^{\imath \theta}) r^3 e^{-r^2} dr d\theta
\end{equation}
The question concerns then the convergence of $\int_0^{+\infty} |f(re^{\imath \theta})|r^3e^{-r^2}dr$. Firstly we set $\int_0^{+\infty} |f(re^{\imath \theta})|r^3e^{-r^2}dr = \int_0^{1} |f(re^{\imath \theta})|r^3e^{-r^2}dr + \int_1^{+\infty} |f(re^{\imath \theta})|r^3e^{-r^2}dr$ to separate the problems of convergence at $0$ and at $+\infty$.
\begin{itemize}
\item Firstly, we can note that
  \begin{eqnarray}
    \Delta \varphi_X(\beta) & = & \varphi_X(\beta+\alpha_A) - \varphi_X(\beta) \\
    & = & \sum_{n=1}^{+\infty} \sum_{m=0}^{n} \left. \frac{\partial^n \varphi_X}{\partial z^m \partial \bar z^{n-m}} \right|_{z=\alpha_A} \frac{\beta^m \bar \beta^{n-m}}{m!(n-m)!}
  \end{eqnarray}
  (the sum over $n$ can stop after $\deg(\varphi_X)$ if $\varphi_X$ is polynomials). If $\left. \frac{\partial \varphi_X}{\partial z} \right|_{z=\alpha_A} \not=0$ or $\left. \frac{\partial \varphi_X}{\partial \bar z} \right|_{z=\alpha_A} \not=0$, we have $\Delta \varphi_X \sim r$ for $r$ in the neighbourhood of $0$. So if at least one of the first derivatives of $\varphi_X$ at $\alpha_A$ is non-zero for $X=A$ or $X^3$ (and if at least one of the second derivatives of $\varphi_Y$ at $\alpha_A$ is non-zero if its first derivatives are zero), then $|f| \sim \frac{1}{r^2}$ (or $\sim \frac{1}{r}$) ensuring the convergence of $\int_0^{1} |f(re^{\imath \theta})|r^3e^{-r^2}dr$.
\item For the convergence at $+\infty$, we consider the main situations where $\varphi_X$ are polynomials or, complex or real exponentials:
\begin{itemize}
\item Suppose that $\varphi_A$ and $\varphi_{X^3}$ are polynomials ($X$ and $A$ are polynomials of $a$ et $a^+$). We have  $|f| \sim \frac{1}{r^N}$ (with $N$ bewteen 1 and the double of the highest degree following the considered function $f$). $|f|r^3 e^{-r^2} \sim r^{3-N} e^{-r^2}$ and then $\int_1^{+\infty} |f(re^{\imath \theta})|r^3e^{-r^2}dr$ converges.
\item Suppose that $\varphi_{X^3}$ is polynomials and $\varphi_A(\beta) \propto e^{\bar \gamma \beta - \gamma \bar \beta}$ implying that $|\Delta \varphi_A(\beta)| \propto |e^{\imath \xi r} - 1|$ (with $\xi$ dependent on $\theta$). $|\Delta \varphi_A|$ is bounded and $|f| \sim \frac{1}{r^N}$ (with $N$ between 1 and the double of the degree of $\varphi_{X^3}$). It follows that $\int_1^{+\infty} |f(re^{\imath \theta})|r^3e^{-r^2}dr$ converges.
\item Suppose that $\varphi_A$ is polynomials and $\varphi_{X^3}(\beta) \propto e^{\bar \gamma \beta + \gamma \bar \beta}$ implying that $\Delta \varphi_{X^3}(\beta) \propto e^{\xi r} - 1$ (with $\xi$ dependent on $\theta$). If $\xi <0$, $\Delta \varphi_{X^3}$ is bounded and $|f| \sim \frac{1}{r^N}$. If $\xi >0$, $|f| \sim \frac{e^{-2\xi r}}{r^N}$ (or $\sim \frac{e^{-3\xi r}}{r^N}$). It follows that $\int_1^{+\infty} |f(re^{\imath \theta})|r^3e^{-r^2}dr$ converges.
\item Suppose that $\varphi_A(\beta) \propto e^{\bar \gamma \beta - \gamma \bar \beta}$ and $\varphi_{X^3}(\beta) \propto e^{\bar \gamma' \beta + \gamma' \bar \beta}$ then $|f| \sim e^{-2\xi' r}$ ($\sim e^{-3\xi' r}$, or $\sim 1$), and $\int_1^{+\infty} |f(re^{\imath \theta})|r^3e^{-r^2}dr$ converges.
\end{itemize}
Similar arguments can be used for some other analytical functions $\varphi_X$ (because $X \in \Env(a,a^+)$, $\varphi_X$ is an analytical function of $\beta$ and $\bar \beta$ with a series of infinite convergence radius). 
\end{itemize}

\subsection{Some complex gaussian integrals}
To obtain some expressions of the magnetic potentials it needs to compute some complex integrals. Table \ref{gaussian} provides the values of these ones.
\begin{table}
  \caption{\label{gaussian} Some complex gaussian integrals. $\odd$ denotes any function odd with respect to $\RE\beta$ or to $\IM\beta$. $L,R \in \mathbb R^{+*}$ are two constants.}
  \begingroup
  \renewcommand{\arraystretch}{1.75}
  \begin{tabular}{c|c}
    $f(\beta)$ & $\int_{\mathbb C} f(\beta) |\beta|^2 e^{-|\beta|^2} \frac{d^2\beta}{\pi}$ \\
    \hline
    $\odd$ & $0$ \\
    $\frac{1}{L^2 \RE^2\beta + R^2\IM^2\beta}$ & $\frac{1}{LR}$ \\
    $\frac{\RE^2\beta}{(L^2 \RE^2\beta + R^2\IM^2\beta)^2}$ & $\frac{1}{2L^3R}$ \\
    $\frac{\IM^2\beta}{(L^2 \RE^2\beta + R^2\IM^2\beta)^2}$ & $\frac{1}{2LR^3}$ \\
    $\frac{\RE^2\beta \IM^2\beta}{(L^2 \RE^2\beta + R^2\IM^2\beta)^3}$ & $\frac{1}{8L^3R^3}$ \\
    $\frac{\RE^4\beta}{(L^2 \RE^2\beta + R^2\IM^2\beta)^3}$ & $\frac{3}{8L^5R}$ \\
    $\frac{\IM^4\beta}{(L^2 \RE^2\beta + R^2\IM^2\beta)^3}$ & $\frac{3}{8LR^5}$ \\
    $\frac{\RE^2\beta \IM^2\beta}{(L^2 \RE^2\beta + R^2\IM^2\beta)^2}$ & $\frac{1}{2LR(L+R)^2}$\\
  \end{tabular}
  \endgroup
\end{table}
Note that for functions such that $f(\beta) = \frac{1}{r^2} f(e^{\imath \theta})$ (with $\beta=re^{\imath \theta}$) we have
\begin{equation}
  \int_{\mathbb C} f(\beta)|\beta|^2e^{-|\beta|^2} \frac{d^2\beta}{\pi}  = \underbrace{\frac{1}{\pi} \int_0^{+\infty} re^{-r^2} dr}_{=\frac{1}{2\pi}} \int_0^{2\pi} f(e^{\imath \theta}) d\theta
\end{equation}

\section{Higher dimensions}\label{highdim}
The result presented in this paper can be generalised in higher dimensions by considering a Dirac operator in $\mathbb C^N \otimes \mathscr F$:
\begin{equation}
  \slashed D_x = \Gamma_I \otimes (X^I-x^I)
\end{equation}
with $\tr\Gamma_I=0$. The general model depends on three dimensional parameters:
\begin{itemize}
\item $d$: the dimension of the target space $\mathbb R^d$ of $x$ (number of the $\Gamma$ matrices);
\item $N$: the dimension of the polarisation space $\mathbb C^N$ (size of the $\Gamma$ matrices);
\item $m$: the number of bosonic modes of $\mathscr F$ (number of pair creation/annihilation operators $(a_i,a^+_i)_{i=1,...,m}$; the dimension of the CCR algebra is $2m+1$).
\end{itemize}
For example, in the quantum information theory, the control of a qutrit in contact with a bimode bosonic environment is characterized by $d=8$ ($\{\Gamma_I\}_{I=1,...,8}$ are the Gell-Mann matrices), $N=3$ and $m=2$; and the geometry can be the one of a D4-brane in $\mathbb R^8$.\\
A simplest example can be found in M-theory BFSS model of a D3-brane in $\mathbb R^4$ with a massless fermionic string in Weyl representation with parameters $d=4$, $N=4$ and $m=2$. In that case $\forall i \in \{1,...,3\}$, $\Gamma_i = \gamma_0 \gamma_i$ (with $\gamma_\mu$ the usual Dirac matrices), and $\Gamma_4 = \gamma_0 \gamma_5$:
\begin{equation}
  \slashed D_x = \left(\begin{array}{cc} \sigma_i \otimes (X^i-x^i) & X^4-x^4 \\ -(X^4-x^4) & -\sigma_i \otimes (X^i-x^i) \end{array} \right)
\end{equation}
Suppose that $\mathscr F = \mathscr F_a \otimes \mathscr F_b$ has two modes $(a,a^+)$ and $(b,b^+)$, and that $X^1,X^2,X^3 \in \Env(a,a^+)$ and $X^4 = \frac{b-b^+}{2\imath}$. Let $(|\Lambda_\bullet\rrangle)_{\bullet=\circ,*}$ be the quasicoherent states of $\sigma_i \otimes (X^i-x^i)$ and $|x^4\rangle \in \mathscr F_b$ be $\langle \xi|x^4 \rangle = e^{\imath x^4 \xi}$ with the harmonic oscillator representation where $\mathscr F_b = L^2(\mathbb R,d\xi)$ and $X^4 = -\imath \frac{\partial}{\partial \xi}$. $\slashed D_x$ has eight quasicoherent states:
\begin{equation}
  |\Lambda_{\bullet_1 \bullet_2 \pm}(x) \rrrangle  =  \frac{1}{\sqrt 2} \left(\begin{array}{c} |\Lambda_{\bullet_1}(x^1,x^2,x^3) \rrangle \\ \pm |\Lambda_{\bullet_2}(x^1,x^2,x^3) \rrangle \end{array} \right) \otimes |x^4 \rangle
\end{equation}
$\bullet_i \in \{\circ,*\}$. The eigensurface is $M_\Lambda \times \mathbb R$ (where $M_\Lambda$ is the eigensurface of $\sigma_i \otimes (X^i-x^i)$).\\

Returning to the general case $\slashed D_x = \Gamma_I \otimes (X^I-x^I)$, we can rewrite the main steps of the demonstration of the theorem \ref{Theo}. Let $[\slashed D_x]_{ab} = A_{ab}-\alpha_{ab}$ with $A_{ab} \in \mathcal L(\mathscr F)$, $\alpha_{ab} \in \mathbb C$, $A_{ba} = A_{ab}^\dagger$ and $\alpha_{ba} = \bar \alpha_{ab}$. Let $|\Lambda \rrangle = (|\Lambda^1,...,|\Lambda^N\rangle)^t$ with $|\Lambda^a\rangle \in \mathscr F$.
\begin{eqnarray}
  \forall a, & \quad & (A_{ab}-\alpha_{ab}) |\Lambda^b \rangle = 0 \\
  & \iff & (\varphi_{ab}(\beta) - \alpha_{ab}) \langle \beta|\Lambda^b \rangle = \sum_{i=1}^m c_{ai} (\beta_i-\gamma_{ai}) \langle \beta|\gamma_{a} \rangle
\end{eqnarray}
where $\beta=(\beta_i,...,\beta_m) \in \mathbb C^m$, $|\beta \rangle = |\beta_1 \rangle \otimes ... \otimes |\beta_m \rangle$ ($|\beta_i\rangle$ being the coherent state of the mode $i$: $a_i|\beta_i \rangle = \beta_i|\beta_i\rangle$), and with any $c_{ai} \in \mathbb C$. $\varphi_{ab} : \mathbb C^m \to \mathbb C$ is the $|\alpha\rangle$-representation of $A_{ab}$ ($\varphi_{aa} : \mathbb C^m \to \mathbb R$ with $A_{aa}$ self-adjoint). To satisfy this equation in the cases where $\beta_i = \gamma_{ai}$ it needs that
\begin{equation}
  \gamma_{a} \in \bigcap_{b=1}^N \varphi^{-1}_{ab}(\alpha_{ab})
\end{equation}
The equation $\slashed D_x |\Lambda \rrangle = 0$ has the following solutions:
\begin{equation}
  |\Lambda_{ai}\rrangle = \int_{\mathbb C^m} \langle \beta+\gamma_a|\Lambda_{ai}\rrangle \otimes |\beta+\gamma_a \rangle \frac{d^2\beta^1...d^2\beta^m}{\pi^m}
\end{equation}
\begin{equation}
  \langle \beta|\Lambda_{ai}^b \rangle = (\beta_i-\gamma_{ia}) \langle \beta|\gamma_a\rangle {[\mathcal D_x(\beta)^{-1}]^b}_a
\end{equation}
with $\mathcal D_x(\beta)_{ab} =  \varphi_{ab}(\beta) - \alpha_{ab}$, for each $\gamma_a \in \bigcap_b \varphi^{-1}_{ab}(\alpha_{ab})$. The number of independent solutions depends strongly on $\bigcap_{b=1}^N \varphi^{-1}_{ab}(\alpha_{ab})$ which can be large if $m$ is large.

\end{document}